\documentclass[journal, 12pt, onecolumn]{IEEEtran}
\usepackage{cite}
\usepackage{amsmath,amssymb,amsfonts, amsthm}
\usepackage{algorithmic}
\usepackage{graphicx}
\usepackage{textcomp}
\usepackage{subcaption}
\usepackage{eric}
\usepackage{algorithm}
\usepackage{xcolor}

\def\BibTeX{{\rm B\kern-.05em{\sc i\kern-.025em b}\kern-.08em
    T\kern-.1667em\lower.7ex\hbox{E}\kern-.125emX}}
\begin{document}
\bstctlcite{IEEEexample:BSTcontrol}
\title{Neural Estimation of the Rate-Distortion Function With Applications to Operational Source Coding}%and Operational Source Coding}
\author{
\IEEEauthorblockN{Eric Lei, Hamed Hassani, and Shirin Saeedi Bidokhti} \\ 
\IEEEauthorblockA{Department of Electrical and Systems Engineering \\ University of Pennsylvania, Philadelphia, PA}
\thanks{Journal submission under review. This work was presented in part at the IEEE International Symposium on Information Theory (ISIT), 2022. 
The work of Eric Lei was supported by a NSF Graduate Research Fellowship. The work of Shirin Saeedi Bidokhti was supported by NSF award 1910594 and an NSF CAREER award 2047482. The work of H. Hassani was supported by NSF award CIF-1943064. This work was also partially supported by the AI Institute for Learning-Enabled Optimization at Scale (TILOS), NSF award CCF-2112665. Email: \{elei, hassani, saeedi\}@seas.upenn.edu. \textit{(Corresponding author: Eric Lei.)}
% This paragraph of the first footnote will contain the date on 
% which you submitted your brief for review. It will also contain support 
% information, including sponsor and financial support acknowledgment. For 
% example, ``This work was supported in part by the U.S. Department of 
% Commerce under Grant BS123456.'' 
}
% \thanks{Email: \{elei, hassani, saeedi\}@seas.upenn.edu.}
}

\maketitle

\begin{abstract}
A fundamental question in designing lossy data compression schemes is how well one can do in comparison with the rate-distortion function, which describes the known theoretical limits of lossy compression. Motivated by the empirical success of deep neural network (DNN) compressors on large, real-world data, we investigate methods to estimate the rate-distortion function on such data, which would allow comparison of DNN compressors with optimality. While one could use the empirical distribution of the data and apply the Blahut-Arimoto algorithm, this approach presents several computational challenges and inaccuracies when the datasets are large and high-dimensional, such as the case of modern image datasets. Instead, we re-formulate the rate-distortion objective, and solve the resulting functional optimization problem using neural networks. We apply the resulting rate-distortion estimator, called NERD, on popular image datasets, and provide evidence that NERD can accurately estimate the rate-distortion function. Using our estimate, we show that the rate-distortion achievable by DNN compressors are within several bits of the rate-distortion function for real-world datasets. Additionally, NERD provides access to the rate-distortion achieving channel, as well as samples from its output marginal. Therefore, using recent results in reverse channel coding, we describe how NERD can be used to construct an operational one-shot lossy compression scheme with guarantees on the achievable rate and distortion. Experimental results demonstrate competitive performance with DNN compressors.
\end{abstract}

\section{Introduction}

    Driven by advances in deep neural network (DNN) compression schemes, rapid progress has been made in finding high-performing lossy compression schemes for large, high-dimensional datasets that remain practical \cite{Balle2017, Theis2017a, SoftToHardVQ, NTC}. While these methods have empirically shown to outperform classical compression schemes for real-world data (e.g. images), it remains unknown as to how well they perform in comparison to the fundamental limit, which is given by the rate-distortion function. To investigate this question, one approach is to examine a stylized data source with a known probability distribution that is analytically tractable, such as the sawbridge random process, as done in \cite{Wagner2021NeuralNO}. This allows for a closed-form solution of the rate-distortion function; one can then compare it with empirically achievable rate and distortion of DNN compressors trained on realizations of the source. However, this approach does not evaluate DNN compressors on true sources of interest, such as real-world images, for which architectural choices such as convolutional layers have been engineered \cite{LecBen15}. Thus, evaluating the rate-distortion function on these sources is paramount to understanding the efficacy of DNN compressors on real-world data. 
    
    Furthermore, a class of information-theoretically designed one-shot lossy source codes with near-optimal rate-distortion guarantees, which fall under the area of reverse channel coding \cite{sfrl, theis2021algorithms, harsha, entanglement, DCS, quantumreverse, bravermangarg, cuff2008, poissonmatching}, can provide a one-shot benchmark for DNN compressors, which are typically one-shot. However, these schemes require the rate-distortion-achieving conditional distribution (see \eqref{eq:RD}), which is generally intractable for real-world data, especially when the data distribution is unknown and only samples are available. Having the ability to recover the rate-distortion function's optimizing conditional distribution only from samples, in addition to the rate-distortion function itself, would allow for implementation of reverse channel codes even without access to the full data distribution. 
    
 Consider an independent and identically-distributed (i.i.d.) data source $X \sim P_X$, where $P_X$ is a probability distribution supported on alphabet $\mathcal{X}$. Let $\mathcal{Y}$ be the reproduction alphabet, and $\dist: \mathcal{X} \times \mathcal{Y} \rightarrow \mathbb{R}^+$ be a distortion function on the input and output alphabets. The asymptotic limit on the minimum number of bits required to achieve a distortion $D$ is given by the rate-distortion function \cite{shannon48, shannonRD, CoverThomas},  defined as
 %   \begin{definition} The rate-distortion function $R(D)$ of a source $P_X$ under distortion function $\dist$ is given by 
        \begin{equation}
    	    R(D) := \inf_{\substack{P_{Y|X}:  \EE_{P_{X,Y}}[\dist(X,Y)]  \leq D}} I(X;Y), 
    	    \label{eq:RD}
    	\end{equation} 
   % \end{definition} 
    Any rate-distortion pair $(R,D)$ satisfying $R > R(D)$ is achievable by some lossy source code, and no code can achieve a rate-distortion less than $R(D)$. It is important to note that $R(D)$ is achievable only under asymptotic blocklengths, whereas DNN compressors are typically one-shot, as compressing i.i.d. blocks for real-world datasets may not be feasible. However, the one-shot achievable region is known to be within $\log(R(D)+1)+O(1)$ bits of $R(D)$ \cite{sfrl}, and thus  even in the one-shot setting, $R(D)$ remains an appropriate measure of the fundamental limits. 
    
     \label{sec:BA}
    \begin{figure}[t]
         \centering
        \includegraphics[width=0.45 \textwidth]{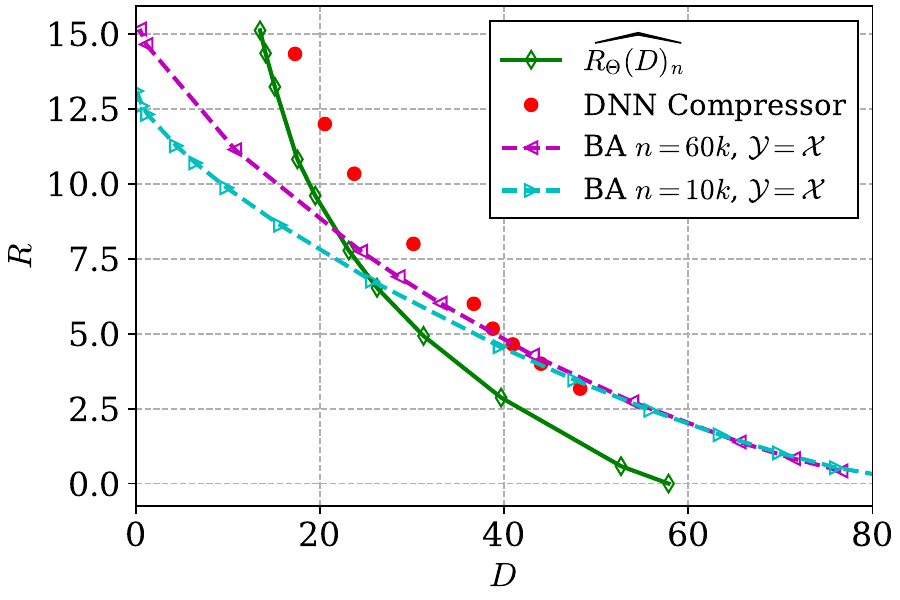}
        \caption{Inaccuracy of discretized Blahut-Arimoto in comparison to our method, $\widehat{R_\Theta(D)}_n$, for computing the rate distortion curve on the MNIST dataset. DNN compressors provide codes that lie in the achievable region. See text for details.}
        \label{fig:plug_in}
    \end{figure}
    
    There are several immediate challenges when computing $R(D)$ (and its optimizing conditional distribution) for large-scale data. Even when the distribution of $P_X$ is known, the analytical form of the rate-distortion function has been difficult to evaluate, and has been characterized only on specific sources. This prohibits an analytical derivation using a density estimate of $P_X$ (which are also not sample-efficient in high dimensions) in most cases. Computational methods such as the Blahut-Arimoto (BA) algorithm seem to be better suited for our setting; however, as will be shown, BA provides inaccurate estimates and fails to scale with large datasets.
    
    \subsection{Blahut-Arimoto Fails to Scale}
    Let $\DKL(\mu || \nu)$ be the Kullback-Leibler (KL) divergence, defined as $\EE_\mu\bracket{\log \frac{d\mu}{d\nu}}$ when the density $\frac{d\mu}{d\nu}$ exists and $+\infty$ otherwise. Due to the convex and strictly decreasing properties \cite{CoverThomas} of $R(D)$, it suffices to fix some $\beta > 0$, and solve the following double optimization problem.

    \begin{lemma}[Double-Minimization Form, cf. {\cite[Ch.~10]{CoverThomas}}, {\cite{Yeung2002AFC}}]
        The minimizers $P^{(\beta)}_{Y|X}$, $Q_Y^{(\beta)}$ of
        \begin{equation}
	     \RD(\beta):=\inf_{Q_Y} \inf_{P_{Y|X}   } \DKL(P_{X,Y}|| P_X \otimes Q_Y) + \beta \mathop{\EE}_{P_{X,Y}}[\dist(X,Y)],
	     \label{eq:double}
	    \end{equation}
	    yield a unique point $R_\beta = \DKL(P_XP^{(\beta)}_{Y|X} || P_X \otimes Q_Y^{(\beta)})$ and $D_\beta = \EE_{P_X P^{(\beta)}_{Y|X}} [\dist(X,Y)]$ on the positive-rate regime of the rate-distortion curve, i.e. $R(D_\beta) = R_\beta$, such that $D_\beta < D_{\mathrm{max}}$ where $R(D_{\mathrm{max}}) = 0$.
        % \begin{proof}
        % See \cite[Lemma~10.8.1]{CoverThomas}.
        % \end{proof}
        % \begin{proof}
        %     For the joint measure defined by $P_{X,Y}=P_X P_{Y|X}$ on $(\mathcal{X} \times \mathcal{Y}, \mathcal{A} \otimes \mathcal{B})$, we have that
        %     \begin{equation}
        %         \inf_{Q_Y \in \mathcal{P}(\mathcal{Y})} \DKL(P_{X,Y}|| P_X \otimes Q_Y) = \DKL(P_{X,Y} || P_X \otimes P_Y)
        %     \end{equation}
        %     which follows by the steps in the proof of Lemma~10.8.1 of \cite{CoverThomas} which also hold under the more abstract setting.
        %     We can thus rewrite \eqref{eq:lagrangian} as
        %     \begin{align}
        %         &R(D) = \inf_{\substack{P_{Y|X}:  \mathop{\EE}_{P_{X,Y}}[\dist(X,Y)]} \leq D} \DKL(P_{X,Y} || P_X \otimes P_Y)  \nonumber \\
        %         &= \inf_{\substack{P_{Y|X}:  \mathop{\EE}_{P_{X,Y}}[\dist(X,Y)]} \leq D} \inf_{Q_Y} \DKL(P_{X,Y}|| P_X \otimes Q_Y) \nonumber \\
        %         &= \inf_{Q_Y} \inf_{\substack{P_{Y|X}:  \mathop{\EE}_{P_{X,Y}}[\dist(X,Y)]} \leq D} \DKL(P_{X,Y}|| P_X \otimes Q_Y) 
        %     \end{align}
        % \end{proof}
    \end{lemma}
	
	 The Blahut-Arimoto (BA) solves \eqref{eq:double} by alternating steps on $P_{Y|X}$ and $Q_Y$ until convergence.  In discrete settings,  the optimizers take the following closed form:
        \begin{align}
            p(y|x) &= \frac{r(y)e^{-\beta \dist(x,y)}}{\sum_{\tilde{y} \in \mathcal{Y}} r(\tilde{y})e^{-\beta \dist(x,\tilde{y})}}, \quad \forall x \in \mathcal{X}, y \in \mathcal{Y} \label{eq:BA_1} \\
            r(y) &= \sum_{x \in \mathcal{X}} p_X(x) p(y|x), \quad \forall y \in \mathcal{Y}
        \end{align}
    Even though BA requires knowledge of the source distribution, one can use the empirical distribution $\hat{P}_n = \frac{1}{n}\sum_{i=1}^n \delta_{X_i}$ as a proxy. This, however, does not scale in the case of modern datasets. Consider the setting when $\mathcal{X} = \mathbb{R}^m$ is continuous. Applying BA requires discretization of the input and output alphabets. In many cases, this would require acute knowledge of how to discretize $\mathbb{R}^m$ to form an appropriate reconstruction alphabet $\mathcal{Y}$, and even if one could, it might result in computational complexity that grows, potentially exponentially, with $m$. One would need to store a $n \times |\mathcal{Y}|$ matrix for the conditional PMFs and a $|\mathcal{Y}|$-sized vector for the output marginal PMF, which may not fit in memory depending on the number of data points or the choice of discretization. For example, in image compression, where we assume each $X_i \in \mathbb{R}^m$ to be a single image realization, $\mathcal{Y} \subseteq \mathbb{R}^m$. Even for 8-bit grayscale images, full precision quantization would require $2^8 \cdot m$ points, and although one could provide better discretization schemes, they may still require an intractable number of points. 
    
    To demonstrate, we attempt to apply discretized BA to MNIST digits in Fig.~\ref{fig:plug_in}, and plot its estimated curve in comparison to rate-distortion (with squared-error distortion) achieved by DNN compressors. Specifically, our source is the empirical MNIST distribution $\hat{P}_n$, and we discretize $\mathcal{Y}$ to be exactly the support of our data, i.e. $\mathcal{Y} = \bigcup_{i=1}^n \{X_i\}$. While this scheme should converge to the true rate-distortion function as $n \rightarrow \infty$ \cite{RDplugin}, we see that even for $n=60,000$, this fails to capture the general trend of the DNN compressors. Finally, the  distortion corresponding to $R=0$, known as $D_{\max}$, that BA estimates is far from the optimal given by $\EE_{P_X}\bracket*{\|X-\mu_X\|_2^2}$ -- see Section~\ref{sec:results} for more details. This showcases the inaccuracy of BA in estimating the rate-distortion function even with relatively large number of samples. In contrast, our method, which provides the estimate $\widehat{R_\Theta(D)}_n$, does not exhibit these failures and is able to generalize to the true MNIST distribution for a significant portion of the rate-distortion function.

    \subsection{Related Work}
    There have been several recent works that attempt to estimate $R(D)$ on real-world image data. Aside from the classical works of Arimoto \cite{arimoto} and Blahut \cite{blahut}, the first work that uses neural networks to estimate rate-distortion is \cite{qing2020}, who parameterize the $Q_{Y|X}$ channel using restricted Boltzmann machines and study small synthetic sources. More recently, there have been works such as \cite{yang2021towards, yang2021lower, gibson17}, in which the authors bound the rate-distortion function on real-world data. The authors of \cite{yang2021towards} provide sandwich bounds on $R(D)$, where the upper bound is variational and proved to be tight. In contrast, this work, which was independently developed around the same time, provides an estimate of $R(D)$ by replacing a class of distributions that $R(D)$ minimizes over with a parameterized set of distributions, leading to a natural upper bound. Additionally, the works in \cite{yang2021towards, theis2021algorithms, flamich2020compressing} discuss the potential of reverse channel coding applied to image compression; our work directly implements reverse channel coding using the rate-distortion achieving $Q_{Y|X}$ channel learned from our $R(D)$ estimate. In \cite{RDplugin}, the authors analyze theoretical properties of the plug-in estimator for $R(D)$, but do not provide a method that can be applied to real-world datasets.
    
    A related area of work lies in the generative modeling literature, where the rate-distortion trade-off is often used to evaluate generative models and unsupervised learning algorithms. The most relevant work is \cite{Huang2020EvaluatingLC}, where the authors take a rate-distortion perspective to evaluate the performance of generative adversarial networks (GANs) and variational autoencoders (VAEs). In their formulation, they assume the trained generative model is the output $\mathcal{Y}$-marginal of the rate-distortion objective, and find an upper bound on the rate-distortion needed to reproduce the generative model, not the true rate-distortion function of the \textit{source}. Much of the other work in this area \cite{Burgess2018UnderstandingDI, Brekelmans2019ExactRI, brokenELBO} use variational bounds on the rate-distortion for the purposes of representation learning, and lack a direct connection to  fundamental limits of lossy compression.

    \subsection{Contributions} 
    As opposed to the aforementioned approaches in Sec~\ref{sec:BA}, we take a step back and reformulate the rate-distortion objective into a min-max objective using duality, building on results from \cite{Dembo}. As will be shown, this alleviates many of the issues plaguing previous methods, and allows for practical implementation of lossy compression schemes based on reverse channel coding, solely from samples.
    Our contributions are as follows. 
    \begin{itemize}
        \item We propose an estimator of $R(D)$ based on neural networks, called NERD, which we show is strongly consistent, and provide a corresponding algorithm to compute the estimate from samples for a broad class of distortion measures. 
        \item We empirically show that these methods provide accurate  estimates of $R(D)$ on synthetic as well as real-world datasets, and that DNN autoencoder compressors achieve a rate-distortion within a few bits of our estimate.
        \item We demonstrate how the optimal conditional distribution (or channel) of $R(D)$ can be approximately recovered from NERD, and applied to reverse channel coding schemes, which result in an operational one-shot lossy compression scheme.
        \item We experimentally show that on real-world data, this scheme performs competitively with DNN compressors while also satisfying guarantees on the achievable rate and distortion. 
        \item We provide evidence that the gap between the one-shot DNN compressors and the estimated rate-distortion function could be minimized by using DNN compressors that perform block coding.
    \end{itemize}
    
    % In essence, we provide a scalable method to estimate the rate-distortion function on real-world datasets, and provide evidence that it is accurate using known properties of $R(D)$; we also demonstrate the achievable rate-distortion of state-of-the-art (one-shot) DNN compressors to be within several bits of our estimate and competitive. These findings open up further avenues for  research. In particular, it remains open whether the gap between the performance of DNN compressors and the asymptotic $R(D)$ function is due to their one-shot nature and if we could  close this gap by developing DNN compressors that perform block coding.  It also remains open as to whether or not our proposed one-shot codes described in Sec.~\ref{sec:one-shot} will work in practice and could outperform DNN compressors. 
    
    \subsection{Notation}
    \label{sec:notation}
    We use $\EE[\cdot]$ and $\PP(\cdot)$ to denote expectation and probability, respectively. In general, we use subscript letters to denote a probability measure's respective space, e.g. $Q_Y$ for a distribution supported on $\mathcal{Y}$. The distribution $P_X$ refers to the source (or data) distribution, supported on $\mathcal{X}$. For a measure $\mu$, we use $f_* \mu$ to denote the pushforward measure of $\mu$ through a function $f$. We use $\otimes$ to denote product measures, e.g. $\mu \otimes \nu$. We assume logarithms to be taken base 2. $\dist(\cdot, \cdot)$ represent a distortion measure, $\DKL(\cdot || \cdot)$ is the Kullback-Leibler divergence, and $W_p(\cdot, \cdot)$ is the $p$-Wasserstein distance. 
    
    \section{Problem Formulation}
    \label{sec:pf}

Our goal is to estimate the rate-distortion function $R(D)$ of some source $P_X$. However, we only have access to $n$ samples $X_1,\dots,X_n$ drawn i.i.d. from $P_X$, and do not assume any other knowledge about its distribution.

    As opposed to BA, which solves the double minimization problem \eqref{eq:double} in closed form, we use a dual form of the rate-distortion function. We first note that the constrained version of the inner minimization problem in \eqref{eq:double} is known as the \textit{rate function} in the literature \cite{RDplugin, Dembo}, i.e.
    \begin{equation}
        R(Q_Y, D) := \inf_{\substack{P_{Y|X}: \\ \mathop{\EE}_{P_{X,Y}}[\dist(X,Y)] \leq D}} \DKL(P_{X,Y}|| P_X \otimes Q_Y), 
        \label{eq:rate_func}
    \end{equation}
    which exhibits the following dual characterization.
    \begin{lemma}[Rate Function Duality, {\cite[Sec. 2]{Dembo}}] The rate function can be equivalently expressed as follows. 
    \begin{equation}
        R(Q_Y, D) = \sup_{\tilde{\beta} \leq 0} \,\, \tilde{\beta} D - \EE_{P_X} \left[\log \EE_{Q_Y} \left[e^{\tilde{\beta} \dist(X,Y)}\right] \right] 
    \end{equation}
    \end{lemma}
    % \begin{proof}
    %             See \cite[Prop.~2]{Dembo}.
    % \end{proof}
    \noindent Therefore, $R(D)$ is equivalent to $\inf_{Q_Y} R(Q_Y, D)$ and can be expressed as a min-max problem,
    \begin{equation}
        R(D) = \inf_{Q_Y} \sup_{\tilde{\beta} \leq 0} \,\, \tilde{\beta} D - \EE_{P_X} \bracket*{\log \EE_{Q_Y} \bracket*{e^{\tilde{\beta} \dist(X,Y)}}},
        \label{eq:doubledual}
    \end{equation}
    which can be estimated from samples, since we can approximate expectations with empirical averages for both $P_X$,  $Q_Y$ independently. Furthermore, the inner problem is concave, scalar, and has a unique solution. To solve the inner max, first-order stationary conditions yield \cite[Sec.~2.2]{Dembo} 
    \begin{equation}
        D = \EE_{P_X \otimes Q_Y} \bracket*{\dist(X,Y) \frac{\exp\paran*{\tilde{\beta}\dist(X,Y)}}{\EE_{Y' \sim Q_Y} \bracket*{\exp\paran*{\tilde{\beta}\dist(X,Y')}}}},
        \label{eq:stationary}
    \end{equation} 
    which can be solved for $\tilde{\beta}^*$ via the bisection method. In the next section, we use the dual formulation to derive an estimator that uses neural networks to parametrize $Q_Y$.
    
    % \begin{remark}
    %     A naive approach is to parametrize $Q_Y$ using neural networks in the double minimization form in \eqref{eq:double}. In fact, the mapping approach \cite{mapping_approach} uses a similar idea, where the space $\mathcal{Z}$ is discretized rather than $\mathcal{Y}$, and $G$ can be optimally fit to find the best discretization of $\mathcal{Y}$. One could then solve the inner minimization in \eqref{eq:double} via the first BA step \eqref{eq:BA_1} using batched samples as uniform empirical distributions, similar to the method in Sinkhorn GANs \cite{sinkhornGAN}. However, doing so still suffers from needing a number of samples exponential in $R$, since the KL term is still computed on a discrete distribution.
    % \end{remark}
	
	\section{Neural Estimation of the Rate-Distortion Function}
 \label{sec:NERD}
    We propose to  parametrize the output marginal distribution $Q_Y$ using architectural choices similar to those used in the GAN literature \cite{GAN}. Specifically, let $P_Z$ be some simple base distribution over $\mathcal{Z}$ and let $G: \mathcal{Z} \rightarrow \mathbb{R}^m$ be a function belonging to a function class $\mathcal{G}$. Then, representing distributions $Q_Y$ with the pushforward $G_* P_Z$, we can optimize over functions in $\mathcal{G}$, and arrive at
	\begin{align}
	    \hspace{-0.5em} R_{\mathcal{G}}(D) := \inf_{G \in \mathcal{G}} \sup_{\tilde{\beta} \leq 0} \tilde{\beta} D - \EE_{P_X} \bracket*{\log \EE_{P_Z} e^{\tilde{\beta} \dist(X,G(Z))}}.
	    \label{eq:nrd}
	\end{align}
	
	The equivalence of this (under certain assumptions on $P_Z$) with $R(D)$ is justified in \cite{mapping_approach}. In practice, we only have access to samples $X_1,\dots,X_n$ drawn i.i.d. from $P_X$, and must estimate \eqref{eq:nrd} from the empirical distribution $\hat{P}^{(n)}_X := \frac{1}{n}\sum_{i=1}^n \delta_{X_i}$. Leveraging the expressive power of neural networks, we choose $\mathcal{G}$ to be the class of functions parametrized by neural networks, and arrive at the following estimator (NERD).
	
	\begin{definition}[Neural Estimator of the Rate-Distortion Function (NERD)]
	Let $\mathcal{G} := \{G_\theta\}_{\theta \in \Theta}$ be a class of functions parametrized by a neural network. NERD is given by
	\begin{align}
	    \hspace{-0.5em} \widehat{R_\Theta(D)}_n &:= \inf_{\theta \in \Theta} \sup_{\tilde{\beta} \leq 0} \,\, \tilde{\beta} D - \mathop{\EE}_{\hat{P}^{(n)}_X} \bracket*{\log \mathop{\EE}_{P_Z} \bracket*{e^{\tilde{\beta} \dist(X,G_\theta(Z))}}} \\
	    &= \inf_{\theta \in \Theta} \sup_{\tilde{\beta} \leq 0} \,\, \tilde{\beta} D - \frac{1}{n}\sum_{i=1}^n \log \mathop{\EE}_{P_Z} \bracket*{e^{\tilde{\beta} \dist(X_i,G_\theta(Z))}}
	    \label{eq:NERD}
	\end{align}
% 	where $G_{\theta}\# P_Z$ is the pushforward measure of $P_Z$ by $G_\theta$. 
	\end{definition}
	
	\begin{algorithm}[tb]
       \caption{Neural Estimator of the Rate-Distortion Function (NERD)}
       \label{alg:NERD}
        \begin{algorithmic}
          \STATE {\bfseries Input:} Distortion constraint $D$, batch size $B$, number of steps $T$, learning rate $\eta$
          \STATE Initialize generator neural network $G_\theta : \mathcal{Z} \rightarrow \mathcal{Y}$
        %   \STATE Let $\hat{p}^{(X)} = \frac{1}{B} \vec{1} \in \mathbb{R}^B, ~ \hat{p}^{(Y)} = \frac{1}{B} \vec{1} \in \mathbb{R}^B$
           \FOR{$t=1,2,\dots,T$}
                \STATE Sample $\{x_i\}_{i=1}^B \stackrel{\mathrm{i.i.d.}}{\sim} P_X$
               \STATE Sample $\{z_j\}_{j=1}^B \stackrel{\mathrm{i.i.d.}}{\sim} P_Z$
               \STATE Define $\kappa_{i,j}(\tilde{\beta}, \theta_t) := \exp \paran*{\tilde{\beta}\dist(x_i,G_{\theta_t}(z_j))}$
               \STATE Solve $D = \frac{1}{B}\sum_{i=1}^B \dist(x_i,G(z_i)) \frac{B \cdot \kappa_{i,i}(\tilde{\beta}, \theta_t)}{\sum_{j=1}^B \kappa_{i,j}(\tilde{\beta}, \theta_t)}$ for $\tilde{\beta}^*$
               \STATE{$\theta_{t+1} \leftarrow \theta_t - \eta \nabla_{\theta} \paran*{-\frac{1}{B}\sum_{i=1}^B \log \frac{1}{B} \sum_{j=1}^B \kappa_{i,j}(\tilde{\beta}^*, \theta_t)} $}
           \ENDFOR
    %   \RETURN $\tilde{\beta}D - \frac{1}{B}\sum_{i=1}^B \log \frac{1}{B} \sum_{j=1}^B e^{-\tilde{\beta}\dist(x_i, G_{\theta_T}(z_j))}$
        \end{algorithmic}
    \end{algorithm}

The next theorem shows that NERD is a strongly consistent estimator for the rate distortion function. The proof is provided in Appendix~\ref{sec:proofs}.

    \begin{theorem}[Strong consistency of NERD]
        Suppose the alphabets are $\mathcal{X}= \mathcal{Y} = \mathbb{R}^m$, $P_Z$ is supported on $\mathcal{Z} \in \mathbb{R}^l$, and $P_Z$ is absolutely continuous with respect to Lebesgue measure. Also, suppose that the distortion measure $\dist$ is $L_{\dist}$-Lipschitz in both arguments. Then the NERD estimator in \eqref{eq:NERD} is a strongly consistent estimator of $R(D)$, i.e. 
        \begin{equation}
            \lim_{n \rightarrow \infty} \widehat{R_\Theta(D)}_n = R(D) ~\textrm{almost surely.}
        \end{equation}
        \label{thm:strong_consistency}
    \end{theorem}
    Note that while NERD satisfies strong consistency, this may not necessarily apply to settings encountered in practice. In practice, we use stochastic gradient descent to search over the function class $\Theta$ which may not necessarily find the minimum, and the expectation over $P_Z$ is estimated using Monte-Carlo methods.

      To use NERD, following \eqref{eq:NERD}, one can simply sample batches from $P_X$ and $P_Z$, solve the inner max of \eqref{eq:doubledual}  by solving  \eqref{eq:stationary} for $\tilde{\beta}^*$, and take a gradient step over the DNN parameters. The full algorithm is given in Algorithm~\ref{alg:NERD}. The code is available online\footnote{{\tt https://github.com/leieric/NERD-RCC}}.

      \subsection{Numerical Estimation Challenges}
      \label{sec:challenges}
      Note that although NERD is strongly consistent, the method suffers from estimation inaccuracies for large rates, which mirror similar estimation challenges of mutual information from samples \cite{mcallester}. The issues stem from the log expectation over $P_Z$, $\log \EE_{P_Z} \bracket*{e^{\tilde{\beta} \dist(X_i, G_\theta(Z))}}$, requiring at least $2^R$ samples to estimate accurately with a sample mean at rate $R$. To see this, note that the inner minimization in the min-max form in \eqref{eq:doubledual} is equivalent to the Donsker-Varadhan (DV) lower bound \cite{DV}. From \cite{Dembo}, if $\tilde{\beta^*}$ is the inner problem's maximizer, then 
      \begin{align}
          \sup_{\tilde{\beta} \leq 0} \,\, \tilde{\beta} D - \EE_{P_X} \bracket*{\log \EE_{Q_Y} \bracket*{e^{\tilde{\beta} \dist(X,Y)}}} &=  \mathop{\EE}_{P_X}\bracket*{ \mathop{\EE}_{Q^*_{Y|X}(\cdot|X)} \bracket*{\tilde{\beta^*}\dist(X,Y)} - \log \mathop{\EE}_{Q_Y} \bracket*{e^{\tilde{\beta^*} \dist(X,Y)}}}  \\
          &= \mathop{\EE}_{P_X} \bracket*{\sup_{f: \mathcal{Y} \rightarrow \mathbb{R}} \mathop{\EE}_{Q^*_{Y|X}(\cdot|X)} \bracket*{f(Y)} - \log \EE_{Q_Y} \bracket*{e^{f(Y)}}} \\
          &= \mathop{\EE}_{P_X} \bracket*{\DKL(Q^*_{Y|X}(\cdot|X)||Q_Y)} \\
          &= I(Q_Y, Q^*_{Y|X}),
      \end{align} 
      where $\frac{dQ^*_{Y|X=x}}{dQ_Y} (x, y) = \frac{e^{\tilde{\beta^*}\dist(x,y)}}{\EE_{Y' \sim Q_Y}\bracket*{e^{\tilde{\beta^*} \dist(x,y)}}}$, and the optimizing $f$ is given by $f^*(y) = \lambda \dist(x, y)$. Hence, if we use an empirical distribution for $Q_Y$, say with $M$ samples, then $I(Q_Y, Q^*_{Y|X}) \leq H(Q_Y) = \log_2 M$. Additionally, the authors in \cite{mcallester} show that the DV lower bound as well as any other $M$-sample estimate of a distribution-free high-confidence lower bound of mutual information is at most $O(\log M)$. 
      Therefore, we can expect NERD to estimate $R(D)$ for rates up to $\log_2 M$. In practice, we set $M=40,000$ for 32 x 32 images.

	\section{Experimental Results: NERD}
	\label{sec:results}
    In our experiments, we use synthetic data (i.e. Gaussian), MNIST digits, and Fashion MNIST (FMNIST) images to represent our source $X \sim P_X$. We use squared-error distortion for all cases, i.e. $\dist(x,y) = \|x-y\|_2^2$. In all cases, we have $n=60,000$ i.i.d. samples from $P_X$. We set the base distribution as $P_Z = \mathcal{N}(0, I_{m_z})$ and parametrize $G_\theta : \mathbb{R}^{m_z} \rightarrow \mathbb{R}^m$ with a fully connected neural network for the Gaussian case, and a deep convolutional architecture similar to the generator architecture used in DCGAN \cite{DCGAN} for the image data.  
    
    % Due to numerical instability of the objective due to the $\log$ term, we instead solve the following objective 
    %   \begin{equation}
    %         \label{eq:epsilon_obj}
    %        \inf_{\theta \in \Theta} \sup_{\tilde{\beta} \leq 0} \,\, \tilde{\beta} D - \mathop{\EE}_{\hat{P}^{(n)}_X} \bracket*{\log \paran*{\mathop{\EE}_{P_Z} \bracket*{e^{\tilde{\beta} \dist(X,G_\theta(Z))}}+\epsilon}},
    %   \end{equation}
    %   for some fixed but small $\epsilon > 0$.
    
    \begin{figure}[t]
    \begin{minipage}{0.475\textwidth}
        \centering
        \includegraphics[width=0.95 \textwidth]{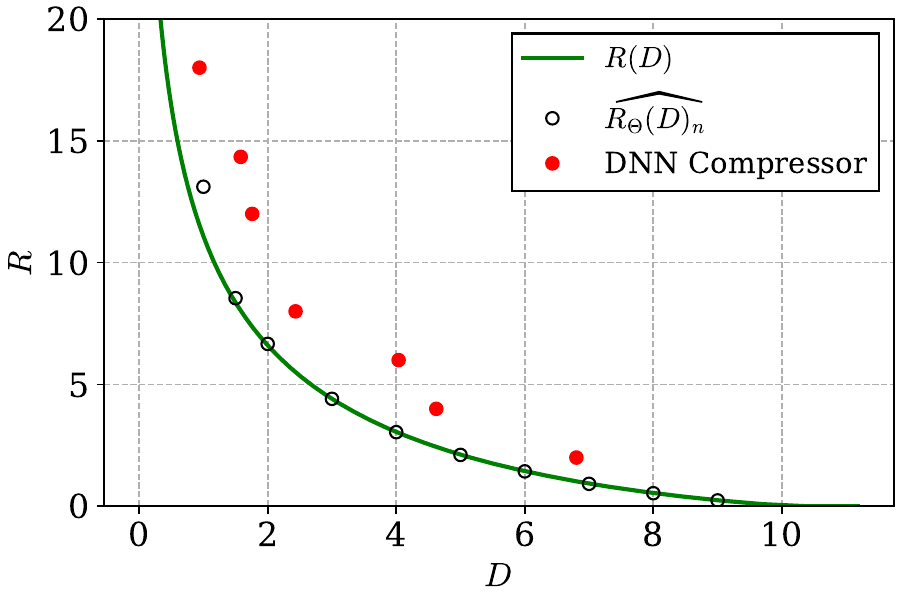}
        \caption{Estimated $\widehat{R_\Theta(D)}_n$ (NERD) on Gaussian data ($m=20$, $\sigma_k^2 = 4e^{-\frac{1}{16}k}$) compared with DNN compressors.}
        \label{fig:gaussian_NERD}
    \end{minipage}
    \hfill
    \begin{minipage}{0.475\textwidth}
        \centering
        \includegraphics[width=0.95 \textwidth]{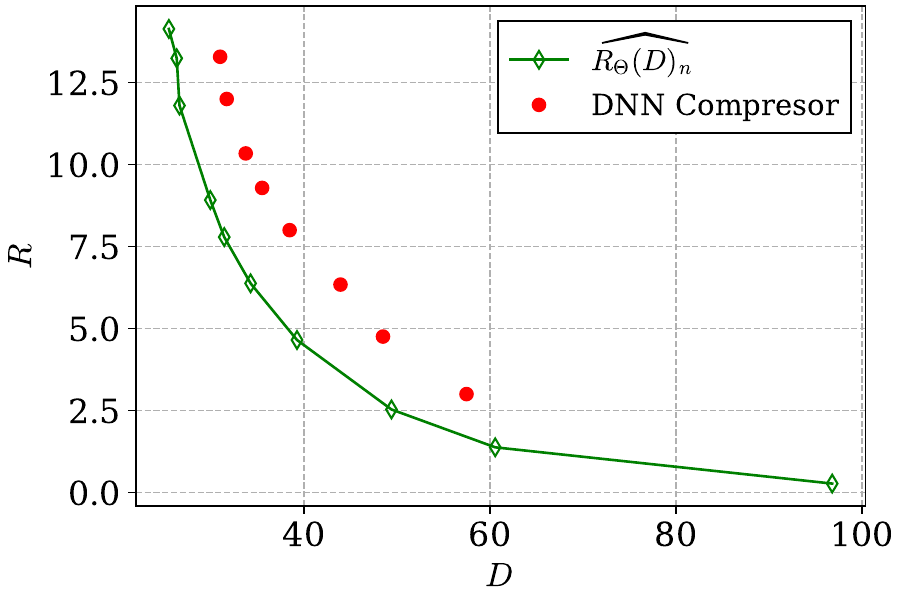}
        \caption{Estimated $\widehat{R_\Theta(D)}_n$ (NERD) of  SVHN images vs. DNN compressors.}
        \label{fig:SVHN_NERD}
    \end{minipage}
    \end{figure}
    
    \subsection{Synthetic Data}
    We first verify the NERD estimator on Gaussian data. Specifically, $P_X = \mathcal{N}(0, \Sigma)$ is a multivariate Gaussian with $m=20$ dimensions. Let $\Sigma = V \mathrm{diag}(\sigma_1^2, \dots, \sigma_{m}^2) V^\top$ be the eigendecomposition of $\Sigma$, with $V$ orthogonal. In this case, the rate-distortion function of $P_X$ has an analytical form and is given by \cite{CoverThomas}
	\begin{equation}
	    \label{eq:gaussian_rd}
	    R(D) = \begin{cases} \sum_{\substack{i \in [m]: \\ \sigma_i^2 > \lambda}} \frac{1}{2} \log \frac{\sigma_i^2}{\lambda}, & D \leq \sum_{i=1}^m \sigma_i^2 \\
	    0, & \text{o.w.}\end{cases}
	\end{equation}
	where $\lambda > 0$ satisfies the reverse water-filling condition
	\begin{equation}
	    \label{eq:rev_wf}
	    \lambda |\{i: \sigma_i^2 > \lambda\}| + \sum_{i \in [m]: \sigma_i^2 \leq \lambda} \sigma_i^2 = D.
	\end{equation}
	To evaluate NERD, we choose $\sigma_k^2 = 4e^{-\frac{1}{16}k}$, and pick an orthogonal matrix $V$ to form $\Sigma$. We generate samples from $\mathcal{N}(0, \Sigma)$ and apply Alg.~\ref{alg:NERD}. As shown in Fig.~\ref{fig:gaussian_NERD}, for this choice of $P_X$, NERD can accurately estimate $R(D)$ for rates below 10 bits. For higher rates, the estimate suffers from numerical instability as discussed in the previous section. We also provide a second Gaussian example in the appendix, shown in Fig.~\ref{fig:gaussian2}.

    \subsection{Real-World Data}
     We apply NERD to MNIST and FMNIST datasets, which are single-channel images, as well as SVHN (Street View House Numbers) \cite{SVHN} which contain color RGB images. We plot the estimated rate-distortion curves in Fig.~\ref{fig:plug_in} for MNIST, Fig.~\ref{fig:RD_only} for FMNIST, and Fig.~\ref{fig:SVHN_NERD} for SVHN. In all cases, the curve appears to satisfy the convex and strictly decreasing properties of the rate-distortion function. Using our estimate of $R(D)$, we can use DNN compressors of the autoencoder type \cite{Balle2017, Theis2017a, SoftToHardVQ} with quantization to see how they perform compared to the fundamental limits. Additionally, we see that DNN compressors closely follow the same trend as the estimated rate-distortion function, and are within several bits of optimality inside the achievable region. However, it remains difficult to conclude in this case whether or not DNN compressors are optimal on these datasets. The gap of several bits could be potentially attributed to the fact that the DNN compressors are one-shot, whereas $R(D)$ is achievable only under asymptotic blocklengths. While one-shot achievable regions of $R > R(D) + \log(R(D)+1)+5$ are known \cite{sfrl}, lower bounds tighter than $R(D)$ for general sources $P_X$ are not as clear. Either way, it remains to be seen whether other computationally feasible source codes could be designed to  perform closer to the rate-distortion limit. In later sections, we discuss learning-based block codes on real-world data that empirically perform closer to $R(D)$ for certain regimes of the rate-distortion tradeoff.
    
%   \textcolor{red}{(see my comments on Skype)} We can interpret this result further. We can interpret DNN compressors as directly attempting to solve the one-shot operational rate-distortion function $R^*(D) := \inf \{R: \exists R\mathrm{-rate~code} ~(f,g)~\mathrm{s.t.}~\EE_{P_X}[\dist(X,g(f(X)))]\leq D \}$, known to be at most $\log R(D)+O(1)$ bits above $R(D)$ \cite{sfrl}, by parametrizing lossy source codes with neural networks. On the other hand, our rate-distortion function estimator attempts to find the \textit{information} rate-distortion function. We know from Shannon's work \cite{shannon48, shannonRD} that they are equivalent under asymptotic blocklengths. Since DNN compressors in the literature implement one-shot codes, the gap of several bits seen in Fig.~\ref{fig:RD_only} can be interpreted as the redundancy, which is known to behave $O(\frac{\log k}{k})$ if $k$ is the blocklength \cite{redundancy,sfrl}. Hence, we hypothesize that there is potential for gains to be made by DNN compressors if one considers neural architectures that compress blocks of samples rather than individual samples.
    
    \subsection{Comparison with Blahut-Arimoto}
    We compare solving \eqref{eq:NERD} to a baseline scheme that uses Blahut-Arimoto on discretized input and output alphabets. On low-dimensional input alphabets, Blahut-Arimoto will perform accurately. But for high-dimensional alphabets, it is unclear how to discretize the continuous space. For image datasets, which are high-dimensional, we have $\mathcal{X} = \{X_1,\dots,X_n \in [0,1]^m\}$ and let the source PMF be $\frac{1}{n}\sum_{i=1}^n \delta_{X_i}(x)$ and choose a discretization for $\mathcal{Y} \subseteq [0,1]^m$ to define an output marginal PMF for Blahut-Arimoto. We attempt to choose the discretization for $\mathcal{Y}$ to be the same as the source, i.e. $\mathcal{Y} = \{X_1,\dots,X_n \in [0,1]^m\}$ is exactly the support of the samples. Such a scheme should converge to the true rate-distortion function as $n \rightarrow \infty$ assuming the true continuous alphabets are both $[0,1]^m$. However, we demonstrate that Blahut-Arimoto fails when the number of samples is finite. Firstly, we are limited by the number of samples (60,000 at most with both datasets). Even with a large number of samples, we see that, given in Fig.~\ref{fig:plug_in} and \ref{fig:RD_only}, doing so does not work particularly well, and the trend is completely off compared to NERD and the DNN codes. It fails to extrapolate to the true rate-distortion function of the true source, and traces the rate-distortion curve for the discrete uniform empirical distribution which we see achieves zero distortion at $R=H(\hat{P}^{(n)}_X) = \log_2 n$. As $n$ grows larger, we would expect the curve traced by this scheme to ``rotate'' clockwise to the true rate-distortion curve (which requires infinite rate at zero-distortion for continuous sources), but this scheme can only rotate to where the zero-distortion rate reaches $\log_2(60,000) \approx 15.87$.  In contrast, NERD is able to follow the same trend of the operational rate-distortion curve estimated by DNN compressors, and matches known characteristics of $R(D)$ as described in the next section. 
    
    \begin{figure}[t]
    \begin{minipage}{0.45\textwidth}
        \centering
        \begin{subfigure}{0.95\textwidth}
             \centering
            \includegraphics[width=0.75\textwidth]{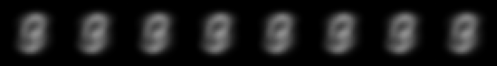}
            \caption{$R=0.28$ bits, $D=54.95$.}
            \label{fig:R0}
        \end{subfigure}
        \\
        \vspace{0.5em}
        \begin{subfigure}[b]{0.95\textwidth}
             \centering
            \includegraphics[width=0.75\textwidth]{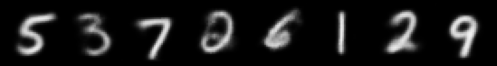}
            \caption{$R = 16.82$ bits, $D = 10$.}
            \label{fig:R8}
        \end{subfigure}
        \caption{Samples from trained $Y$-marginal $Q^*_Y$ (MNIST). As $R \rightarrow 0$, $Q^*_Y$ generates the mean image, achieving $D=D_{\max}$.}
    \end{minipage}
    \hfill
    \begin{minipage}{0.45\textwidth}
        \centering
        \begin{subfigure}{0.95\textwidth}
             \centering
            \includegraphics[width=0.75\textwidth]{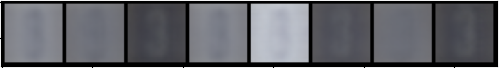}
            \caption{$R=1.37$ bits, $D=60.55$.}
        \end{subfigure}
        \\
        \vspace{0.5em}
        \begin{subfigure}[b]{0.95\textwidth}
             \centering
            \includegraphics[width=0.75\textwidth]{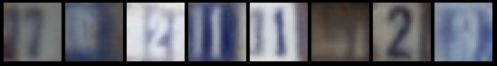}
            \caption{$R = 11.80$ bits, $D = 26.60$.}
        \end{subfigure}
        \caption{Samples from trained $Y$-marginal $Q^*_Y$ (SVHN). As $R \rightarrow 0$, $Q^*_Y$ generates the mean image, achieving $D=D_{\max}$.}
        \label{fig:rates_SVHN}
    \end{minipage}
    \end{figure}

    \subsection{Comparison with Empirical Sandwich Bounds}
    \begin{figure}[t]
        \centering
        \includegraphics[width=0.5\linewidth]{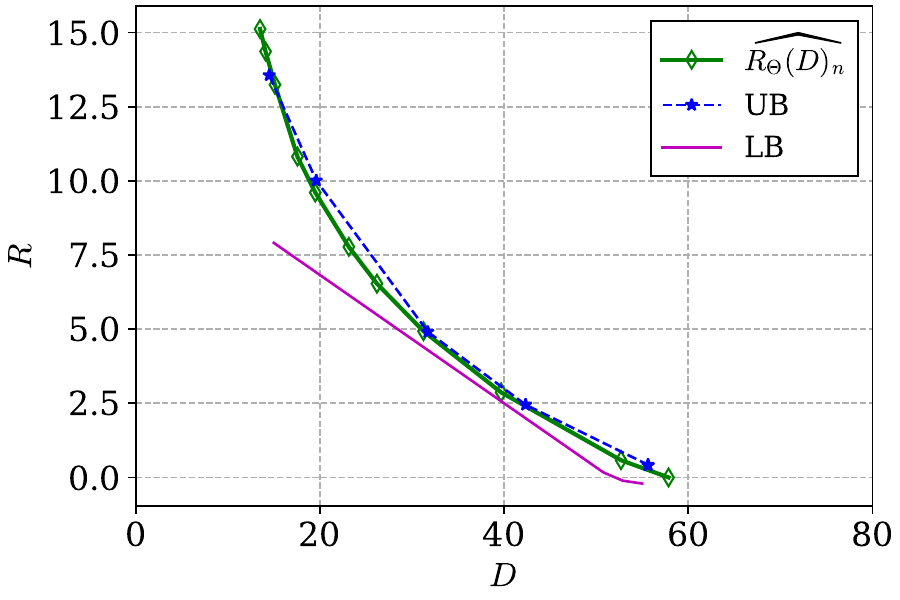}
        \caption{NERD vs. empirical sandwich bounds \cite{yang2021towards} on MNIST. }
        \label{fig:bounds}
    \end{figure}

    This section provides comparisons with other neural network-based bounds on $R(D)$; namely, those provided in \cite{yang2021towards}. It is important to note that NERD itself is an upper bound of $R(D)$. This can be seen since the inner maximization simply computes the mutual information between $Q_Y$ and $Q_{Y|X}$ in closed form. The outer mininimzation restricts the search over distributions $Q_Y$ to those parameterized by the neural network function class. Thus, any fixed $Q_Y$ parameterized by a neural network directly yields an upper bound of $R(D)$. In contrast, the upper bound in \cite{yang2021towards} is a variational bound. However, \cite{yang2021towards} show that it is tight. In Fig.~\ref{fig:bounds} we show how the sandwich bounds provided in \cite{yang2021towards} compare to NERD on the MNIST dataset. Indeed, it can be seen that on certain parts of the rate-distortion curve, NERD matches up with the variational upper bound, indicating that the tightness of their bound is likely achieved for these rate points. For the lower bounds, there is a gap when compared to the NERD estimate. However, we encountered instability from the implementation provided from \cite{yang2021towards}, when computing the lower bound. It is interesting to note that NERD's formulation is very similar to the lower bound \cite[Eq.~4]{yang2021towards}, which suggests two things. First, NERD chooses a simpler parameterization to solve $R(D)$. Second, our discussion from Sec.~\ref{sec:challenges} explained why any lower bound of $R(D)$ requires sample complexity exponential in the rate, which corroborates the analysis and results of the lower bound in \cite{yang2021towards}, who struggled to bring the lower bound above log of the number of samples used.

    \subsection{Samples from the Optimal Reproduction Distribution}
    We now illustrate generated samples from the (approximately) optimal reproduction distribution $Q_Y^*$, parametrized by the trained neural network $G_{\theta^*}$, where $\theta^*$ neural network parameters that are the minimizers of NERD. In other words, $Q_Y^* = G_{\theta^* *} P_Z$, the pushfoward of $P_Z$ through $G_{\theta^*}$. We show that it indeed aligns with the behavior of the rate-distortion function. Let $D_{\max}:= \min_{y \in \mathcal{Y}} \EE_{P_X}[\dist(X, y)]$ be the distortion achievable at zero rate \cite[Ch.~9]{Yeung2002AFC}, i.e. $R(D_{\max}) = 0$. This is the best distortion that can be possibly achieved when there is no information about the source, where the reproduction is simply the best constant estimate of $X$. When $\dist$ is squared-error, and $\mathcal{X}=\mathcal{Y}$, the best constant estimate is the mean  $\mu_{X}=\EE_{P_X}[X]$, with $D_{\max} = \EE_{P_X}\bracket*{\|X-\mu_{X}\|_2^2}$.  In the MNIST case, the samples generated from the generator neural network at $R=0.28$, shown in Fig.~\ref{fig:R0}, consistently generate the ``mean image''. Computing $D_{\max}$ with an empirical average turns out to be $\approx 55$ (under mild preprocessing of the MNIST dataset), which matches the zero-rate point in Fig.~\ref{fig:plug_in}. In contrast, samples generated from a trained generator at higher rate, shown in Fig.~\ref{fig:R8}, appear more similar to the original MNIST images and produce more modes of the distribution. A similar phenomenon occurs with SVHN and FMNIST as well, shown in Fig.~\ref{fig:rates_SVHN} and \ref{fig:rates_fmnist}. In comparison, Blahut-Arimoto does not produce this phenomenon and fails to intersect the $D$-axis at $D_{\mathrm{max}}$. 

    \subsection{Generality of Distortion Function}
    In this section we demonstrate that the method reliably works for general distortion functions. In order for the backpropagation operation to function correctly, the only requirement is that the distortion function be differentiable in its arguments. To demonstrate this, we estimate the rate-distortion function of the same Gaussian source as described previously, but instead of squared-error distortion, we now apply general squared $\ell^p$ distances $\dist_p(x,y) = \|x-y\|_p^2$. Due to the fact that $\|v\|_p \leq \|v\|_q$ when $0 < q \leq p < \infty$, we have that 
    \begin{equation}
        \inf_{\substack{P_{Y|X}:  \EE_{P_{X,Y}}[\dist_p(X,Y)]  \leq D}} I(X;Y) \leq \inf_{\substack{P_{Y|X}:  \EE_{P_{X,Y}}[\dist_q(X,Y)]  \leq D}} I(X;Y).
        \label{eq:RD_ineq}
    \end{equation} 
    We show the rate-distortion function of the Gaussian source using $\dist_1(x,y)$, $\dist_2(x,y)$, and $\dist_3(x,y)$. Since an analytical form of the Gaussian rate-distortion function is not known for $\dist_1$ and $\dist_3$, we only plot the NERD-estimated rate-distortion functions here. 
    \begin{figure}
        \centering
        \includegraphics[width=0.5\linewidth]{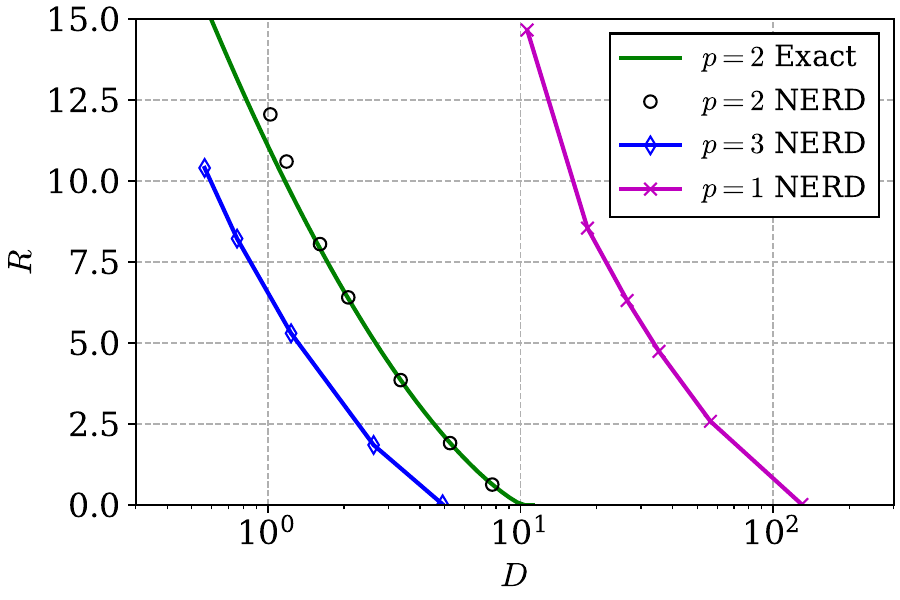}
        \caption{Gaussian rate-distortion function for $\dist_p(x,y) = \|x-y\|_p^2$ for $p=1,2,3$.}
        \label{fig:gaussian_various}
    \end{figure}
    As observed in Fig.~\ref{fig:gaussian_various}, the inequality in \eqref{eq:RD_ineq} is reflected in the 3 rate-distortion curves.

    % \section{Upper Bounds}
    \section{One-Shot Operational Lossy Source Codes}
    \label{sec:RCC}
    Now that we have the capability to calculate the fundamental limit of lossy source coding, and also recover an approximately optimal reproduction distribution $Q_Y^*$ parametrized by a neural network, a natural question to ask is whether it is possible to use $Q_Y^*$ to construct an operational compressor. Indeed, we will see in this section that $Q_Y^*$ can be used to construct a compression scheme with guarantees on the achievable rate and distortion, and empirically perform similar to those of DNN compressors. We first discuss reverse channel coding, which is the technique that underlies the lossy compression scheme, unlike DNN compressors which directly model the encoder and decoder with neural networks.
    
    \subsection{Reverse Channel Coding}
    % \textcolor{red}{(rewrite this part, maybe add a figure, maybe add other names that are used for this problem (mostly because we are submitting to a journal on IT and people have previously worked on this problem. also make sure the only requirement you mention is not that the sample is of dist QY by is jointly distribution with x according to $Q_{Y|X=x}$))}
    The one-shot reverse channel coding (RCC) problem, described in \cite{theis2021algorithms}, consists of reproducing a sample $x \sim P_X$ but allowing it to be corrupted by some noise. It is also known as the channel simulation problem and has been investigated previously in \cite{sfrl, harsha, entanglement, DCS, quantumreverse, bravermangarg, cuff2008, poissonmatching} and references therein. Precisely, given a realization $x \sim P_X$, the sender wishes to reproduce a sample $y$ at the receiver such that it follows a prespecified conditional distribution $Q_{Y|X=x}$, as if the realization $x$ had gone through a channel $Q_{Y|X}$ (see Fig. \ref{fig:RCC}). The sender communicates a message $M$ which the decoder uses to reproduce $y$, with the goal of ensuring $y \sim Q_{Y|X=x}$ and that the expected description length $\EE[L(M)]$ per source symbol, known as rate,  communicated from the sender to receiver is minimal. In our setting, we assume that the sender and receiver share unlimited common randomness, denoted as $U$. 
    
    The relationship to lossy source coding is as follows. In lossy source coding, we wish to approximately represent some sample $x \sim P_X$ with $y$ such that both expected description length and expected distortion are minimized. While RCC is not explicitly concerned with distortion, it is clear that if one has a RCC scheme for channel $Q_{Y|X}$ with rate $\EE[L(M)] \leq R$, then the expected distortion incurred is $\EE_{P_X Q_{Y|X}}[\dist(X,Y)]$. Hence, such a RCC scheme would yield a lossy one-shot code that achieves a rate of $R$ and distortion $\EE_{P_X Q_{Y|X}}[\dist(X,Y)]$. Furthermore, suppose that $Q^*_{Y|X}$ is the channel that minimizes the rate-distortion function \eqref{eq:RD} at distortion level $D$. Any RCC scheme on $Q^*_{Y|X}$ achieves an expected distortion $\EE_{P_X Q^*_{Y|X}}[\dist(X,Y)] \leq D$. What about the rate? It turns out that there are RCC schemes \cite{sfrl, theis2021algorithms} on $Q^*_{Y|X}$ that are guaranteed to achieve rates
    \begin{equation}
        \label{eq:one_shot_RD}
        \EE[L(M)] \leq R(D) + \log (R(D)+1) + 5,
    \end{equation}
    which we now describe.
    
    \begin{figure}
        \centering
        \includegraphics[width=0.65\textwidth]{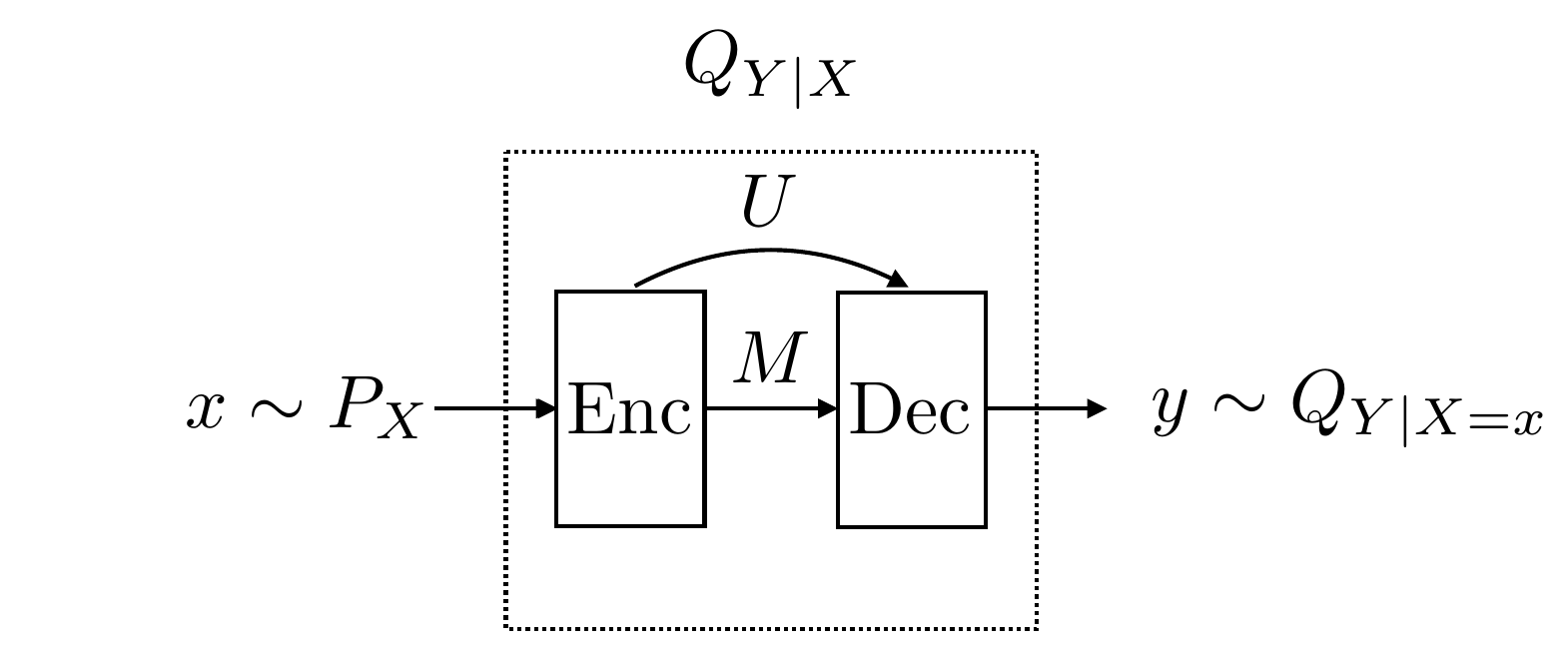}
        \caption{Diagram of reverse channel coding, or channel simulation. Given $x \sim P_X$, the goal is to generate $y \sim Q_{Y|X=x}$, effectively simulating the channel $Q_{Y|X}$, by transmitting some information $M$ with minimal description length $\EE[L(M)]$. Common randomness $U$ between the encoder and decoder is assumed.}
        \label{fig:RCC}
    \end{figure}
    
    \subsubsection{PFR}
    Suppose we have a channel $Q_{Y|X}$. The Poisson functional representation (PFR), proposed by Li \& El Gamal \cite{sfrl}, is one such RCC scheme with guarantees on the rate required. In particular, in order to transmit a sample $x \sim P_X$, both the encoder and decoder (assuming shared common randomness $U$) first sample a marked Poisson process $\{(T_i, Y_i)\}_{i =1}^\infty$, such that $T_i - T_{i-1} \stackrel{\textrm{i.i.d.}}{\sim} \mathrm{Exp}(1)$ and $Y_i \stackrel{\textrm{i.i.d.}}{\sim} Q_Y$, where $Q_Y$ is the marginal of $Y$ over the joint distribution $(X,Y) \sim P_X Q_{Y|X}$. The encoder then computes an index 
    \begin{equation}
        K = \argmin_{i \in \mathbb{N}} T_i\cdot \frac{dQ_Y}{dQ_{Y|X}(\cdot|x)}(Y_i),
        \label{eq:PFR_obj}
    \end{equation} 
    and encodes it using a lossless source code. The decoder recovers $K$, and outputs $Y_K$, which is distributed with $Q_{Y|X=x}$. This scheme has the guarantee that 
    \begin{equation}
        H(K) \leq I(X;Y) + \log(I(X;Y)+1)+4.
        \label{eq:PFR_rate}
    \end{equation}
    In practice, one can encode $K$ using a Huffman codebook designed for the Zipf distribution with PMF $q(k) \propto k^{-(1+1/(I(X;Y)+e^{-1}\log e + 1))}$, which guarantees a rate equivalent to the rate in \eqref{eq:PFR_rate} plus 1 bit \cite{sfrl}. Applying PFR to the rate-distortion optimal channel $Q_{Y|X}^*$ thus achieves the rate guarantee in \eqref{eq:one_shot_RD}.
    
    One potential issue in the practical implementation of the PFR is solving \eqref{eq:PFR_obj}, which requires solving a discrete and infinite optimization problem. To practically implement this, one must use a finite number of samples, but as described in \cite{theis2021algorithms}, the guarantee on the rate in \eqref{eq:PFR_rate} does not necessarily hold when a finite number of samples is used. 
    \subsubsection{ORC}
    To alleviate this issue, Theis \& Yosri \cite{theis2021algorithms} propose ordered random coding (ORC). Rather than weighting the density ratios in \eqref{eq:PFR_obj} by Poisson arrival times, ORC weights them with sorted exponential random variables. 
    The two RCC methods, PFR and ORC, can be summarized as follows \cite[Thm.~3]{theis2021algorithms}. Fix some number of samples $N$. Sample $\{X_i\}_{i=1}^N \stackrel{\textrm{i.i.d.}}{\sim} \mathrm{Exp}(1)$, and $\{Y_i\}_{i=1}^N \stackrel{\textrm{i.i.d.}}{\sim} Q_Y$, where $Q_Y$ is defined as above. Then, the encoder generates the cumulative weights $\{W_i\}_{i=1}^N$, defined to be 
    \begin{equation}
    \label{eq:weights}
        W_i = \begin{cases} \sum_{j=1}^i X_j, & \textrm{if PFR} \\ \sum_{j=1}^i \frac{N}{N-j+1} X_j, & \textrm{if ORC}\end{cases}
    \end{equation} The encoder sends 
    \begin{equation}
        \label{eq:index}
        K = \argmin_{1 \leq i \leq N} W_i \cdot \frac{dQ_Y}{dQ_{Y|X}(\cdot|x)}(Y_i),
    \end{equation}
    and the decoder outputs $Y_K$. In contrast with PFR, ORC can be shown to achieve the same rate as PFR (given in \eqref{eq:PFR_rate}) \cite[Cor. 1]{theis2021algorithms}, but the rate guarantee still holds for some finite $N$ \cite[Thm. 3]{theis2021algorithms}. This guarantees the rate in \eqref{eq:PFR_rate} even in practical usage, and so applying ORC to $Q_{Y|X}^*$ will again achieve the rate in \eqref{eq:one_shot_RD}. 
    
    \subsection{Lossy Compression Scheme via Reverse Channel Coding}
    For any expected distortion tolerance $D$, let $Q_{Y|X}^*$ be the rate-distortion achieving channel, and $Q_Y^*$ be the rate-distortion achieving reproduction distribution (simply the $Y$-marginal of the joint $Q_{Y|X}^* P_X$). Applying PFR or ORC to $Q_{Y|X}^*$ achieves a rate no more than $R(D) + \log(R(D)+1)+5$ and expected distortion no more than $D$. The following proposition exemplifies a nice interpretation when the rate-distortion achieving distributions are used for RCC.
    \begin{proposition}
        \label{RCCenc}
        When using $Q_{Y|X}^*$ and $Q_Y^*$ for RCC to compress $x \sim P_X$, the encoder computes the index 
        \begin{equation}
            K = \argmin_{i} \{\dist(x, Y_i) - (\beta^*)^{-1} \ln W_i \}
        \end{equation}
        where $\beta^* < 0$ is the slope of $R(D)$ at the point $D = \EE_{P_X Q_{Y|X}^*}[\dist(X,Y)]$, and $W_i$'s are generated via \eqref{eq:weights}.
    \end{proposition}
    \begin{remark}
    Prop.~\ref{RCCenc} can be interpreted as the encoder searching for the sample $Y_i$ that is closest to $x$ under distortion measure $\dist$, while simultaneously regularizing the size of the index used (and therefore its entropy). The amount of regularization is proportional to the tradeoff between rate and distortion in the rate-distortion function $R(D)$.
    \end{remark}
    \begin{proof}
        We have that $K = \argmin_i W_i \frac{dQ^*_Y}{dQ_{Y|X}^*(\cdot | x)}(Y_i)$. Let $P_{X,Y} = Q^*_{Y|X} P_X$ be the joint measure. The optimal density ratio is then given by \cite{Dembo} as
        \begin{align}
            \frac{dQ_{Y|X}^*(\cdot | x)}{dQ^*_Y}(y) &= \frac{dP_{X,Y}}{dQ^*_Y dP_X}(y) \\ &= \frac{e^{\beta^*\dist(x,y)}}{\EE_{Y'\sim Q_Y^*}[e^{\beta^*\dist(x,Y')}]}
        \end{align} 
        where $\beta^* = \argmax_{\tilde{\beta} \leq 0} \tilde{\beta}D - \EE_{P_X}\bracket*{\log \EE_{Q_Y^*}\bracket*{e^{\tilde{\beta}\dist(X,Y)}}}$. Hence, using \eqref{eq:index}, we have that 
        \begin{align}
            K &= \argmin_i W_i \frac{dQ^*_Y}{dQ_{Y|X}^*(\cdot | x)}(Y_i) = \argmin_i W_i \frac{\EE_{Y'\sim Q_Y^*}[e^{\beta^*\dist(x,Y')}]}{e^{\beta^*\dist(x,Y_i)}} \\
            &= \argmin_i \frac{W_i} {e^{\beta^*\dist(x,Y_i)}} = \argmin_i \ln \frac{W_i} {e^{\beta^*\dist(x,Y_i)}} \\
            &= \argmin_i \{\ln W_i - \beta^*\dist(x,Y_i)\} = \argmin_i \{\dist(x, Y_i) - (\beta^*)^{-1} \ln W_i \}
        \end{align}
        where in the last step we have rescaled the objective by $-\frac{1}{\beta^*} \geq 0$.
    \end{proof}
    
    Therefore, given $n$ samples from $P_X$, one can design a lossy one-shot code with (approximately) rate \eqref{eq:one_shot_RD} and distortion $D$ as follows: 
    \begin{enumerate}
        \item Apply NERD to the samples from $P_X$, which returns a trained neural network $G_{\theta^*}$ and slope $\beta^*$ as solutions of the min-max problem in \eqref{eq:NERD}. Form an approximation to $Q_Y^*$ as the pushforward of $P_Z$ through $G_{\theta^*}$, i.e. $G_{\theta^* *} P_Z$. Sampling from this distribution can be done by first sampling $Z \sim P_Z$, and outputting $G_{\theta^*}(Z)$.
        \item To compress $x \sim P_X$, apply Alg.~\ref{alg:pfr_enc} with $G_{\theta^* *} P_Z$, slope parameter $\beta^*$, rate parameter as the estimate of the rate-distortion function $\widehat{R_{\Theta}(D)}_n$, and some random seed $r$. This returns a bit string $b$.
        \item To decompress, apply Alg.~\ref{alg:pfr_dec} with the same parameters as the previous step. 
    \end{enumerate}
    We will refer to the above learned compression procedure as \textit{NERD-RCC}. 
    
    \begin{algorithm}[tb]
       \caption{Lossy Encoding Scheme via RCC}
       \label{alg:pfr_enc}
        \begin{algorithmic}
          \STATE {\bfseries Input:} Optimal $\mathcal{Y}$-marginal distribution $Q_Y^*$, rate-distortion slope parameter $\tilde{\beta} < 0$, rate parameter $C$, number of samples $N$, random seed $r$, sample to compress $x$
          \STATE Generate weights $\{W_i\}_{i=1}^N$ according to \eqref{eq:weights}
        %   \STATE Generate Poisson arrivals $\{T_i\}_{i=1}^N$, such that $T_i - T_{i-1} \sim \mathrm{Exp}(1)$, with $T_0=0$
          \STATE Sample $\{Y_i\}_{i=1}^N \stackrel{\mathrm{i.i.d.}}{\sim} Q_Y^*$ using $r$
          \STATE Let $K = \argmin_{i=1,\dots,N} \{\dist(x, Y_i) - \tilde{\beta}^{-1} \ln W_i \}$ 
          \STATE Encode $K$ using a Huffman code for the distribution with mass $q(k) \propto k^{-(1+1/(C+e^{-1}\log e + 1))}$, into a bit string $b$
          \RETURN $b$
        \end{algorithmic}
    \end{algorithm}
    
    \begin{algorithm}[tb]
       \caption{Lossy Decoding Scheme via RCC}
       \label{alg:pfr_dec}
        \begin{algorithmic}
          \STATE {\bfseries Input:} Optimal $\mathcal{Y}$-marginal distribution $Q_Y^*$, rate parameter $C$, number of samples $N$, random seed $r$, bit string $b$
          \STATE Decode $K$ using the Huffman code for the distribution with mass $q(k) \propto k^{-(1+1/(C+e^{-1}\log e + 1))}$
          \STATE Sample $\{Y_i\}_{i=1}^N \stackrel{\mathrm{i.i.d.}}{\sim} Q_Y^*$ using $r$
          \RETURN $Y_K$
        \end{algorithmic}
    \end{algorithm}
    
    \subsection{Experimental Results: NERD-RCC}
    We evaluate the RCC schemes on synthetic Gaussian, MNIST, and FMNIST data.
    \begin{figure}[!t]
    \begin{minipage}{0.475\textwidth}
        \centering
        \includegraphics[width=0.95\linewidth]{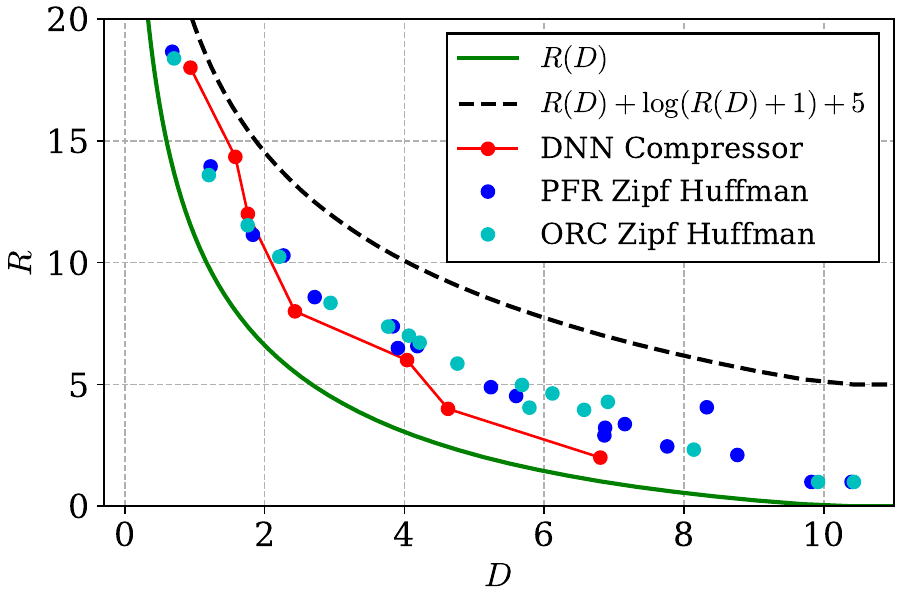}
        \caption{Gaussian source.}
        \label{fig:gaussian_RCC}
    \end{minipage}
    \hfill
    \begin{minipage}{0.475\textwidth}
	     \centering
         \includegraphics[width=0.95\textwidth]{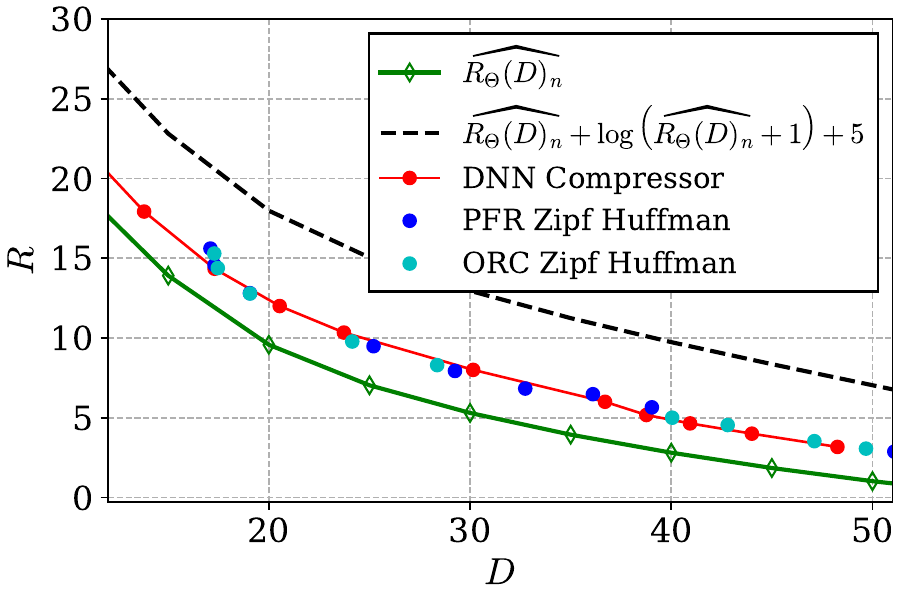}
         \caption{PFR and ORC on MNIST images.}
         \label{fig:pfr_orc_mnist}
	\end{minipage}
	\end{figure}
	
	\begin{figure}[t]
	\begin{minipage}{0.475\textwidth}
	     \centering
         \includegraphics[width=0.95\textwidth]{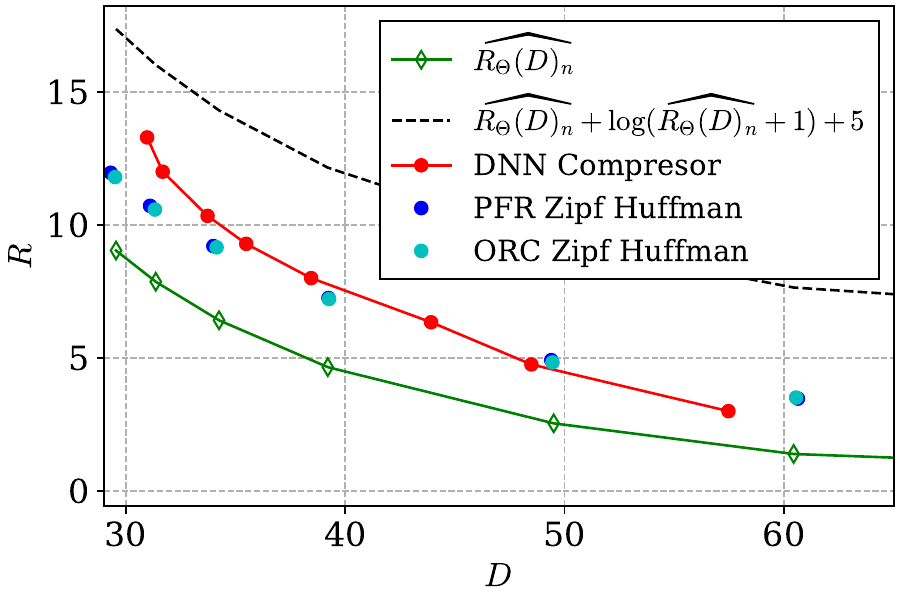}
         \caption{PFR and ORC on SVHN images.}
         \label{fig:pfr_orc_fmnist}
	\end{minipage}
	\hfill
	\begin{minipage}{0.475\textwidth}
        \centering
        \begin{subfigure}{0.95\textwidth}
             \centering
            \includegraphics[width=0.65\textwidth]{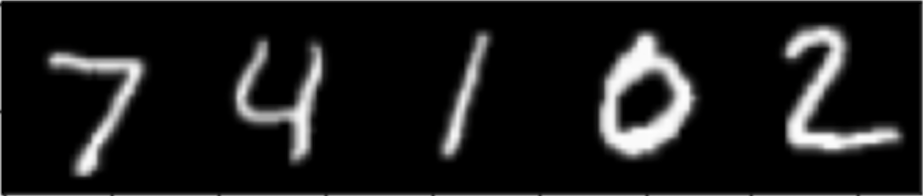}
            \caption{Original MNIST images.}
            \label{fig:PFR_original}
        \end{subfigure}
        \\
        \vspace{0.5em}
        \begin{subfigure}[b]{0.95\textwidth}
             \centering
            \includegraphics[width=0.65\textwidth]{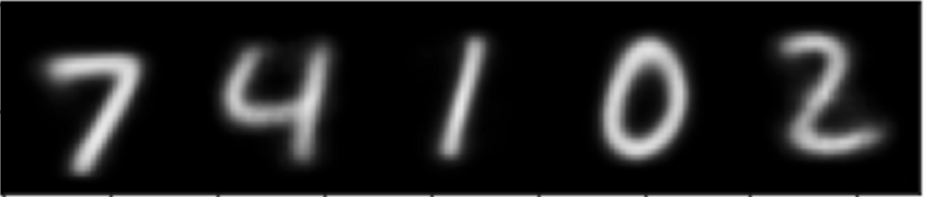}
            \caption{Decompressed MNIST images via PFR.}
            \label{fig:PFR_decompressed}
        \end{subfigure}
        \caption{Visualization of Alg.~\ref{alg:pfr_enc}, \ref{alg:pfr_dec} on MNIST images. }
    \end{minipage}
    \end{figure}
    In the $m$-dimensional Gaussian case, we let the source be $P_X=\mathcal{N}(0, \mathrm{diag}(\sigma_1^2,\dots,\sigma_m^2))$ where $\sigma_k^2 = 4e^{-\frac{1}{16}k^2}$. We can find the rate-distortion achieving channel $Q_{Y|X}^*$ and output marginal distribution $Q_Y^*$ in closed form \cite{CoverThomas}. Let $A = \{k:\sigma_k^2 > \lambda\}$, with $\lambda$ defined in \eqref{eq:rev_wf}. Then the channel and output marginal are both factorized as $Q_{Y|X}^* = \prod_{k=1}^m Q_{Y_k|X_k}$ and $Q_Y^* = \prod_{k=1}^m Q_{Y_k}^*$, where for $k \in A$, $Q_{Y_k|X_k=x} = \mathcal{N}(x, \lambda)$ and $Q_{Y_k} = \mathcal{N}(0, \sigma_k^2 + \lambda)$, and for $k \notin A$, $Q_{Y_k}$ and $Q_{Y_k|X_k=x}$ assign probability 1 to 0. We compare the RCC schemes (implemented using these known $Q_Y^*$ and $Q_{Y|X}^*$) with a one-shot DNN compressor trained on realizations from $P_X$, and show the results in Fig.~\ref{fig:gaussian_RCC}. As can be seen, PFR and ORC achieve comparable rate-distortion to each other, and are several bits below the rate guarantee \eqref{eq:one_shot_RD}. DNN compressors, interestingly, are also several bits within \eqref{eq:one_shot_RD} and seem to perform better than the RCC methods for lower rates but worse at higher rates. In Fig.~\ref{fig:RCC_Gaussian_Both}, we show the same comparison with the rate-distortion achieved by RCC methods via the $Q_Y^*$ learned from NERD. It can be seen that they closely mirror the performance of the RCC methods achieved using the exact $Q_Y^*$.
    
    The MNIST and FMNIST results are shown in Fig. \ref{fig:pfr_orc_mnist} and \ref{fig:pfr_orc_fmnist} respectively. For these datasets, we apply NERD-RCC by using the optimizers $Q_Y^*$ and $\beta^*$ from NERD since the rate-distortion optimizing channel and output marginal cannot be found in closed form. In both cases, NERD-RCC performs very similarly to the DNN one-shot code. As in the Gaussian case, both DNN compressors and NERD-RCC methods are within the rate guarantee of \eqref{eq:one_shot_RD}. In the Gaussian case in Fig.~\ref{fig:gaussian_RCC}, we have access to the true rate-distortion optimizing channel and output marginal, so there is no potential for reduced performance at higher rates due to these factors. Nevertheless, the benefit of NERD-RCC is that the rate achieved at distortion $D$ is always guaranteed to be within the upper bound $\widehat{R_{\Theta}(D)}_n + \log \paran*{\widehat{R_{\Theta}(D)}_n+1} + 5$, no matter what dataset or source samples that NERD uses. In contrast, DNN compressors provide no such performance guarantee. In Fig.~\ref{fig:PFR_decompressed}, we demonstrate several realizations of the source $P_X$ and the output of NERD-RCC. As can be seen, the reverse channel coding scheme outputs a noisy, but faithful reconstruction of the original source. 
    
    \section{Discussion and Open Problems}

    \subsection{Lower Bounds and Block Coding for DNNs}
     The results in the previous two sections demonstrate that one-shot lossy DNN compressors are close to the rate-distortion function of real-world datasets, and competitive with RCC schemes. However, it is difficult to conclude whether or not DNN compressors are optimal, since a characterization of the one-shot fundamental limits are not as clear for general source distributions. While \eqref{eq:one_shot_RD} provides an achievability result, lower bounds are not as clear for general sources. This precludes any definitive answer to optimality of one-shot DNN compressors for real-world datasets.
     
    %  From \cite{sfrl}, if one applies RCC to blocks of $K$ i.i.d. samples at once, rather than in a one-shot manner, the average 
    However, one might ask if perhaps this gap between the estimated rate-distortion function and rate-distortion achieved by the one-shot DNN compressors are due to the one-shot nature of the DNNs, and if DNN compressors that compress blocks of $M$ realizations at a time can help close this gap. To test this, we use a DNN compressor with the same architecture as the one-shot DNN compressor, but apply it to blocks of $M=4$ images combined together to form a larger image, in a $2\times 2$ configuration. We plot the average rate and distortion per image sample in Figs. \ref{fig:block} and \ref{fig:block_fmnist}. 
    \begin{figure}
        \begin{subfigure}{0.475\textwidth}
            \centering
            \includegraphics[width=0.95\linewidth]{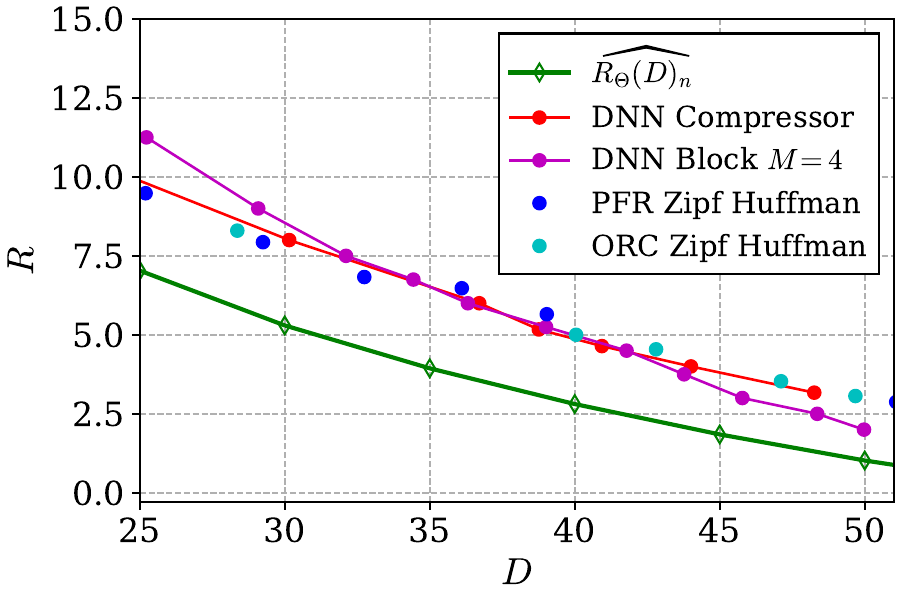}
            \caption{MNIST.}
            \label{fig:block}
        \end{subfigure}
        \begin{subfigure}{0.475\textwidth}
            \centering
            \includegraphics[width=0.95\linewidth]{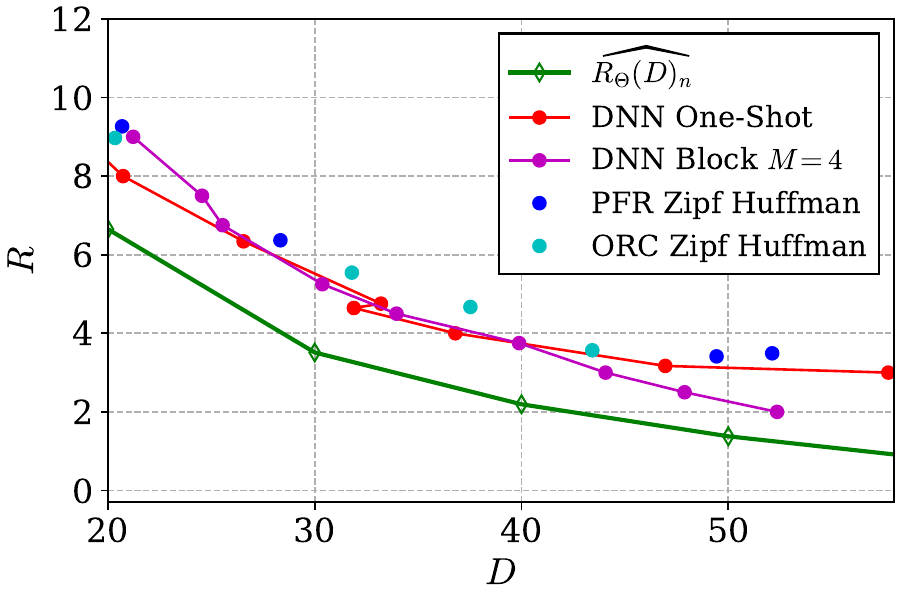}
            \caption{FMNIST.}
            \label{fig:block_fmnist}
        \end{subfigure}
        \caption{DNN block code versus RCC algorithms and DNN one-shot codes.}
    \end{figure}
    As we can see, for both MNIST and FMNIST, the rate-distortion achieved is consistently better than one-shot DNN compressors as well as the RCC algorithms for smaller rates, demonstrating that block DNN codes on i.i.d. sources can indeed provide performance gains. At higher rates, however, the block DNN compressor performs worse. This can potentially be attributed to limitations of the DNN architecture used, which worked well with single images, but may not be the best for 4 images that have been stitched together. One avenue for future work is finding DNN architectures that work well with compressing blocks of images, and to see how they compare to the rate-distortion curves. 
	
	\subsection{Computational Scaling of NERD and RCC}
    \begin{figure}
        \centering
        \includegraphics[width=0.5\linewidth]{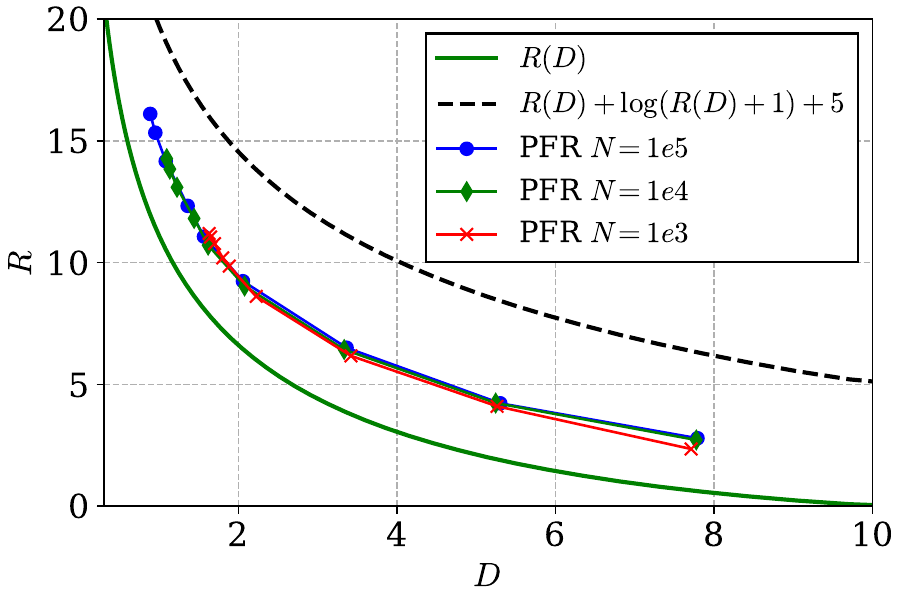}
        \caption{Effect of $N$ on RCC methods. Example here shows PFR on Gaussian source.}
        \label{fig:effectN}
    \end{figure}
    As noted in Sec.~\ref{sec:NERD}, NERD requires a number of samples exponential in the rate required. While this is fine for many real-world datasets such as MNIST and SVHN, this limitation does not allow one to estimate large regimes of the rate-distortion function for ``higher entropy'' datasets such as ImageNet. Designing methods that can accurately estimate $R(D)$ at scale on such datasets is left for future work.
    
    As shown in the previous section, RCC offers an alternative method for learning-based lossy source code design that does not explicitly use quantization, and has guarantees on the achievable rate and distortion which can be estimated directly from data using NERD. However, one drawback of this approach is that the sample complexity of many RCC algorithms is known to scale exponentially with the amount of information communicated from the sender to the receiver \cite{theis2021algorithms, flamich2022fast}. Indeed, when running RCC with varying $N$, the number of samples generated, the estimated rates seem to scale logarithmic with $N$. An example of this is shown in Fig.~\ref{fig:effectN} In comparison, DNN compressors do not suffer from such a runtime complexity. In order to scale to larger rates, one potential avenue is to adapt the recently proposed A$^*$ coding method in \cite{flamich2022fast} for the similar relative entropy coding (REC) problem to the RCC setting. The authors of \cite{flamich2022fast} use A$^*$ coding to achieve sample complexity linear in the rate rather than exponential. Applying these methods to the RCC schemes is left to future work.

\section{Conclusion}
In this paper, we propose a new algorithm for computing the rate-distortion function for real-world data. We use an alternative formulation of the rate-distortion objective which is amenable to parametrization with neural networks and provide an estimator, NERD, that is sample and computationally efficient. We empirically show that it accurately estimates the rate-distortion function for synthetic and real-world datasets. We show how NERD can be used to implement near-optimal compression schemes via reverse channel coding on real-world data, with performance guarantees, and demonstrate the potential of DNN block codes to achieve rate-distortion performance closer to the rate-distortion function.

\appendices

\section{Proofs}
\label{sec:proofs}

    \begin{theorem}[Theorem~\ref{thm:strong_consistency} in text]
        Suppose the alphabets are $\mathcal{X}= \mathcal{Y} = \mathbb{R}^m$, $P_Z$ is supported on $\mathcal{Z} \in \mathbb{R}^l$, and $P_Z$ is absolutely continuous with respect to Lebesgue measure. Also, suppose that the distortion measure $\dist$ is $L_{\dist}$-Lipschitz in both arguments. Then the NERD estimator in \eqref{eq:NERD} is a strongly consistent estimator of $R(D)$, i.e. 
        \begin{equation}
            \lim_{n \rightarrow \infty} \widehat{R_\Theta(D)}_n = R(D) ~\textrm{almost surely.}
        \end{equation}
    \end{theorem}
    \begin{proof}
        Fix $\epsilon > 0$. Define 
        \begin{equation}
            \tilde{R}(Q, \beta):= \beta D - \EE_{P_X} \bracket*{\log \EE_{Q} \bracket*{e^{\beta \dist(X,Y)}}}
        \end{equation}
        where $\beta \leq 0$ is nonpositive, and 
        \begin{equation}
            \tilde{R}(Q) := \max_{\beta \leq 0} R(Q, \beta)
        \end{equation}
        since the inner sup has a unique maximum \cite{Dembo}. By definition $R(D) = \inf_{Q_Y} \tilde{R}(Q_Y)$, there exists $Q_Y^*$ such that $\tilde{R}(Q_Y^*) < R(D) + \epsilon / 4$. We would like to find a distribution $Q_{\theta} \sim G_{\theta}(Z)$, $Z \sim P_Z$, parametrized by a neural network $G_{\theta}$, that has $\tilde{R}(Q_{\theta})$ close to $\tilde{R}(Q_Y^*)$.
        
        For fixed $x$, the function $y \mapsto e^{\beta \dist(x,y)}$ is $L_{\dist} |\beta|$-Lipschitz, since $\beta \leq 0$. Hence, by Lemma~\ref{lemma:W1bound}, 
        \begin{equation}
            \abs*{\EE_{Q_Y^*}\bracket*{e^{\beta d(x,Y)}} - \EE_{Q_{\theta}}\bracket*{e^{\beta d(x,Y)}}} \leq L_{\dist}|\beta| W_1(Q_{\theta}, Q_Y^*)
        \end{equation}
        Let $\beta^*, \beta_{\theta}$ be the maximizers of $\tilde{R}(Q_Y^*)$ and $\tilde{R}(Q_{\theta})$ respectively. Since $\log$ is a continuous function, there exists $\delta_1 > 0$ such that if $W_1(Q_{\theta}, Q_Y^*) < \frac{\delta_1}{L_{\dist} |\beta^*|}$, then 
        \begin{equation}
            \abs*{\log \EE_{Q_Y^*}\bracket*{e^{\beta^* d(x,Y)}} - \log \EE_{Q_{\theta}}\bracket*{e^{\beta^* d(x,Y)}}} < \epsilon / 8
        \end{equation}
        Therefore, choosing $Q_{\theta}$ such that $W_1(Q_{\theta}, Q_Y^*) < \frac{\delta_1}{L_{\dist} |\beta^*|}$ yields
        \begin{align}
            &|\tilde{R}(Q_Y^*, \beta^*) - \tilde{R}(Q_{\theta}, \beta^*)| \\
            &\leq \EE_{P_X}\bracket*{\abs*{\log \EE_{Q_Y^*}\bracket*{e^{\beta^* d(X,Y)}} - \log \EE_{Q_{\theta}}\bracket*{e^{\beta^* d(X,Y)}}}} \\
            &\leq \EE_{P_X} \bracket*{\epsilon/8} = \epsilon / 8
        \end{align}
        where the first inequality holds due to Jensen's inequality. 
        
        We now wish to bound $|\tilde{R}(Q_{\theta}, \beta^*) - \tilde{R}(Q_{\theta}, \beta_{\theta})|$. Since the map $\beta \mapsto \tilde{R}(Q_{\theta}, \beta)$ is continuous in $\beta$, there exists $\delta_2 > 0$ such that if $|\beta^* - \beta_{\theta}| < \delta_2$, then $|\tilde{R}(Q_{\theta}, \beta^*) - \tilde{R}(Q_{\theta}, \beta_{\theta})| < \epsilon/8$. 
        Define
        \begin{equation}
            f(Q,\beta) := \EE_{P_X} \bracket*{\frac{\EE_{Q}\bracket*{\dist(X,Y)e^{\beta \dist(X,Y)}}}{\EE_{Q}\bracket*{e^{\beta \dist(X,Y)}}}} - D
        \end{equation}
        whose root provides the solution to $\tilde{R}(Q)$. We can apply Lemma~\ref{lemma:W1bound} to the numerator and denominator inside the $P_X$ expectation to conclude that with $Q_{\theta}$ satisfying $W_1(Q_Y^*, Q_{\theta}) < C$. Since the function $x_1,x_2 \mapsto \frac{x_1}{x_2}$ is continuous, we can conclude that $f(Q,\beta)$ is continuous with respect to the $W_1$ distance in $Q$ since composition of continuous functions are continuous. Hence there is some $\delta_3 > 0$ such that if $W_1(Q_\theta, Q_Y^*) < \delta_3$, the maximizers $\beta_{\theta}$ and $\beta^*$ must satisfy $|\beta^* - \beta_{\theta}| < \delta_2$, by the implicit function theorem.
        
        We can thus find a neural network $G_{\theta'}$ with sufficient complexity\footnote{Assuming the function class $\Theta$ contains fully-connected networks with ReLU activations, with sufficient width and depth.} via \cite{UAT_distribution} such that $Q_{\theta'}$ satisfies $W_1(Q_{\theta'}, Q_Y^*) < \min \paran*{\frac{\delta_1}{L_{\dist}|\beta^*|}, \delta_3}$, yielding 
        \begin{equation}
            |\tilde{R}(Q_Y^*) - \tilde{R}(Q_{\theta'})| < |\tilde{R}(Q_Y^*, \beta^*)-\tilde{R}(Q_{\theta'}, \beta^*)| + |\tilde{R}(Q_{\theta'}, \beta^*) - \tilde{R}(Q_{\theta'}, \beta_{\theta'})| < \epsilon/4
        \end{equation}
        and hence $|R(D) - \tilde{R}(Q_{\theta'})| < \epsilon/2$.
        Since $\tilde{R}(Q_{\theta'})$ is an upper bound of $R(D)$, we have that 
        \begin{equation}
            |R(D) - R_{\Theta}(D)| = |R(D) - \inf_{\theta \in \Theta} \tilde{R}(Q_{\theta})| < \epsilon / 2
        \end{equation}
        Applying the strong consistency result for the parametric problem in \cite{RDplugin}, $\exists N \in \mathbb{N}$ such that for all $n \geq N$, $|\widehat{R_{\Theta}(D)}_n - R_{\Theta}(D)| < \epsilon/2$ almost surely, and so again by triangle inequality, we have that $|\widehat{R_{\Theta}(D)}_n - R(D)| < \epsilon$ for all $n \geq N$, almost surely.
    \end{proof}
    
    \begin{lemma}
        \label{lemma:W1bound}
        Let $g : \mathcal{X} \rightarrow \mathbb{R}$ be $L$-Lipschitz. Then for any distributions $P,Q \in \mathcal{P}(\mathcal{X})$,
        \begin{equation}
            |\EE_{X\sim P}[g(X)] - \EE_{X \sim Q}[g(X)]| \leq L\cdot W_1(P,Q)
        \end{equation}
        where 
        $W_1(P,Q) := \inf_{\pi \in \Pi(P,Q)} \EE_{X,X'\sim \pi} [\|X-X'\|]$ is the $1$-Wasserstein distance.
    \end{lemma}
    \begin{proof}
        Let $g'(x) = \frac{g(x)}{L}$, so that $g'$ is 1-Lipschitz.
        \begin{align}
            &|\EE_{X\sim P}[g'(X)] - \EE_{X \sim Q}[g'(X)]| \\&\leq \sup_{f: \|f\|_{\mathrm{Lip}} \leq 1}  \EE_{X\sim P}[f(X)] - \EE_{X \sim Q}[f(X)] \\
            &= W_1(P,Q)
        \end{align}
        where the last step is by Kantorovich-Rubinstein duality \cite{KR, ComputationalOT} and $\|f\|_{\mathrm{Lip}} := \sup_{\substack{x_1, x_2 \in \mathcal{X} \\ x_1 \neq x_2}} \frac{|f(x_2)-f(x_1)|}{\|x_1-x_2\|}$ is the Lipschitz norm of $f$.
    \end{proof}

\section{Additional Experiments}

\label{sec:additional_experiments}
In this section, we provide additional experiments to support the main text. Fig.~\ref{fig:gaussian2} provides an additional example of NERD on a different Gaussian source. Figs.~\ref{fig:RD_only}, \ref{fig:pfr_orc_fmnist}, and \ref{fig:rates_fmnist} provides FMNIST results. Fig.~\ref{fig:RCC_Gaussian_Both} provides NERD-estimated $Q_Y^*$ in addition to exact $Q_Y^*$ for RCC methods on the Gaussian source.

\begin{figure}
    \centering
    \includegraphics[width=0.5\linewidth]{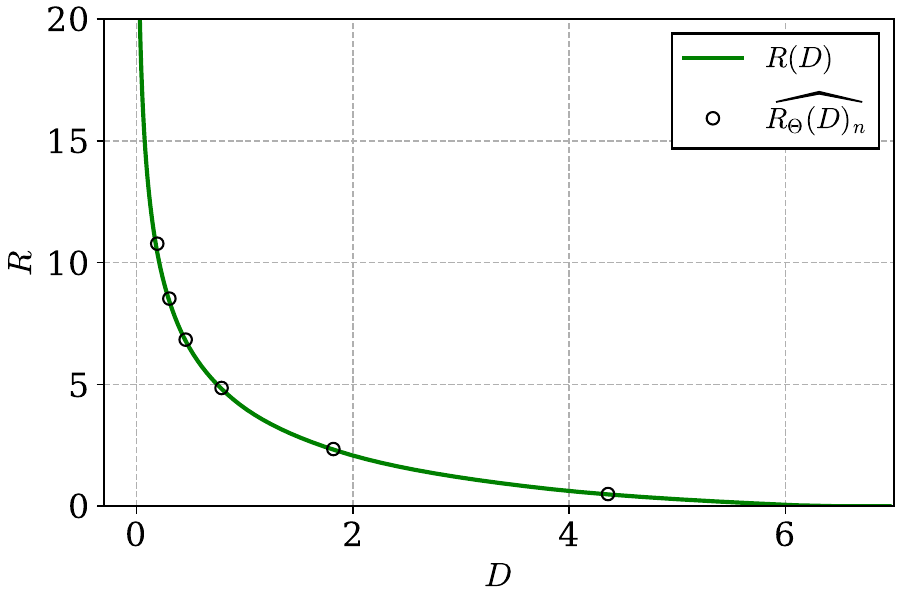}
    \caption{Gaussian rate-distortion with 40 dimensions and $\sigma_k^2 = 4e^{-\frac{1}{4}k}$.}
    \label{fig:gaussian2}
\end{figure}

\begin{figure}[!t]
    \begin{minipage}{0.475\textwidth}
        \centering
    \includegraphics[width=0.95 \textwidth]{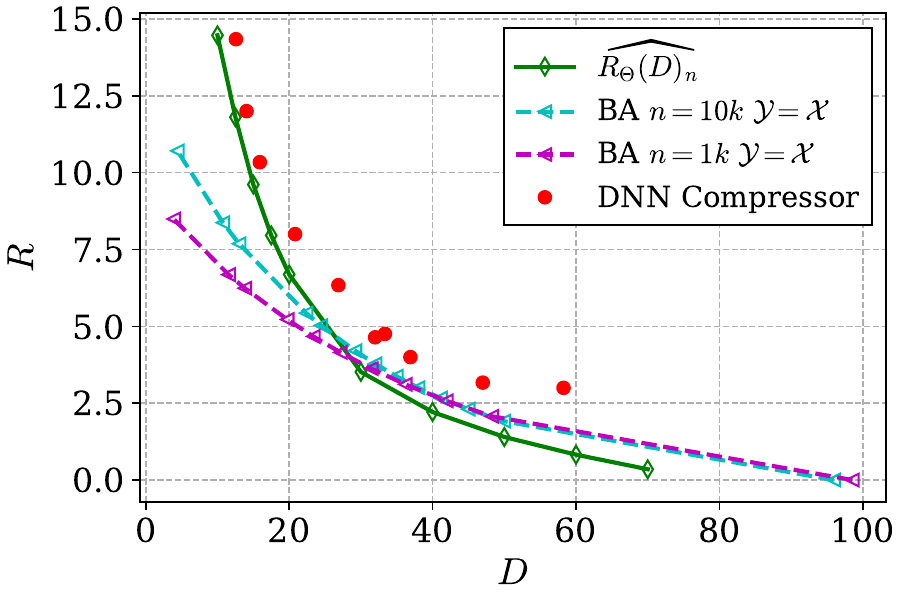}
    \caption{Estimated $\widehat{R_\Theta(D)}_n$ (NERD) of  FMNIST images vs. DNN compressors and BA.}
    \label{fig:RD_only}
    \end{minipage}
    \hfill
    \begin{minipage}{0.475\textwidth}
	     \centering
     \includegraphics[width=.95\textwidth]{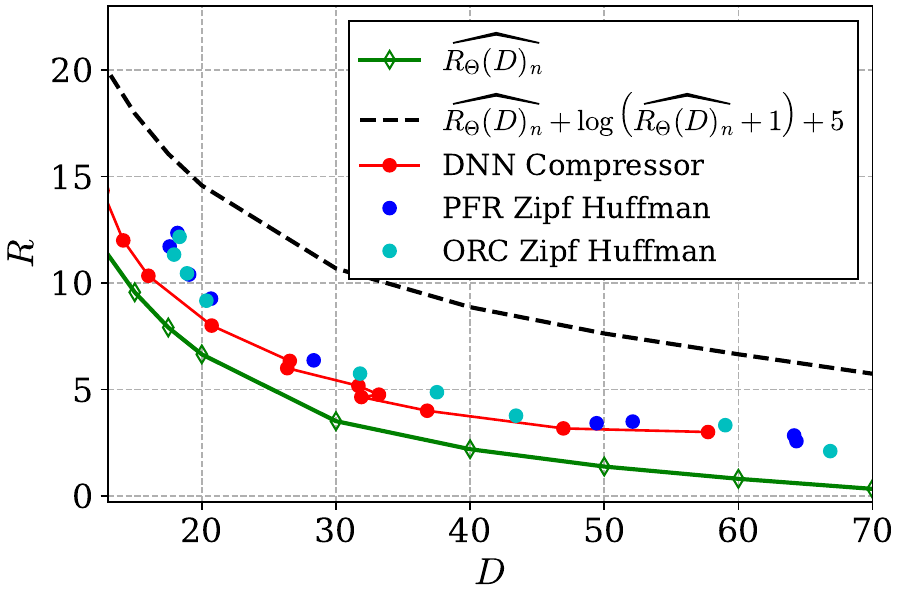}
     \caption{PFR and ORC on FMNIST images.}
     \label{fig:pfr_orc_fmnist}
	\end{minipage}
	\end{figure}

\begin{figure}[!t]
    \begin{minipage}{0.475\textwidth}
        \centering
        \begin{subfigure}{0.95\textwidth}
             \centering
            \includegraphics[width=0.75\textwidth]{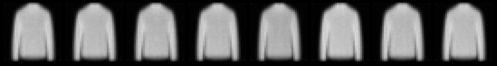}
            \caption{$R=0.33$ bits, $D=70$.}
        \end{subfigure}
        \\
        \vspace{0.5em}
        \begin{subfigure}[b]{0.95\textwidth}
             \centering
            \includegraphics[width=0.75\textwidth]{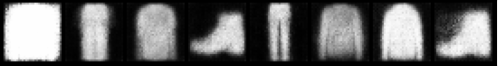}
            \caption{$R = 3.54$ bits, $D = 30$.}
        \end{subfigure}
        \caption{Samples from trained $Y$-marginal $Q^*_Y$ (FMNIST). As $R \rightarrow 0$, $Q^*_Y$ generates the mean image, achieving $D=D_{\max}$.}
        \label{fig:rates_fmnist}
    \end{minipage}
    \hfill
    \begin{minipage}{0.475\textwidth}
	     \centering
        \includegraphics[width=0.95 \textwidth]{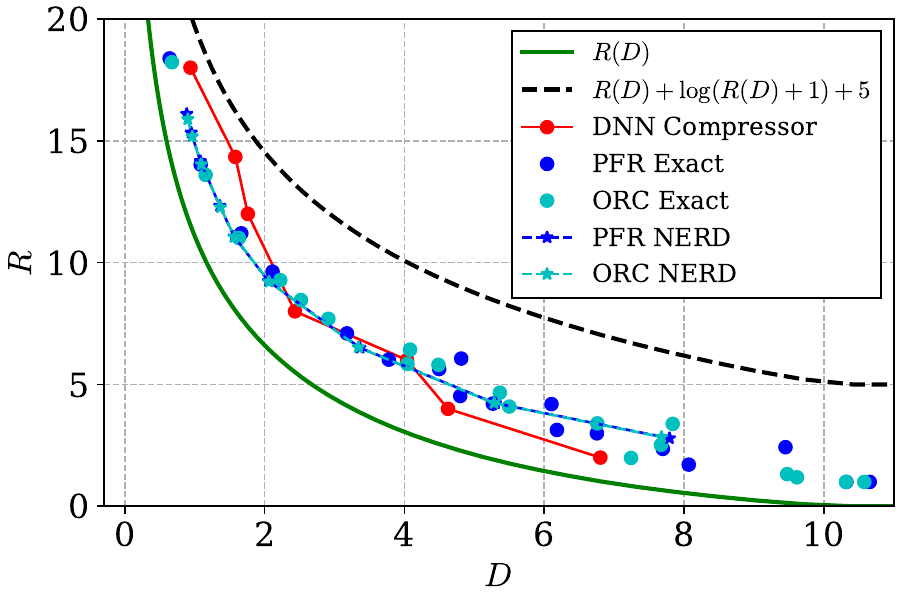}
        \caption{RCC methods on Gaussian source with exact and NERD-estimated $Q_Y^*$.}
        \label{fig:RCC_Gaussian_Both}
	\end{minipage}
	\end{figure}

% \section{\blue{Implementation Details}}
% \label{sec:implementation}

% \section*{Acknowledgment}

% We thank

\bibliographystyle{IEEEtran}
\bibliography{IEEEabrv, ref}

\end{document}